\pdfoutput=1
\documentclass[onecolumn,12pt]{IEEEtran}
\usepackage{enumerate}
\usepackage{setspace}
\usepackage{geometry}
\usepackage{graphicx}
\usepackage[bf,SL,BF]{subfigure}
\usepackage[cmex10]{amsmath}
\newtheorem{Theorem}{Theorem}
\newtheorem{Lemma}{Lemma}
\newtheorem{Definition}{Definition}
\newtheorem{Remark}{Remark}
\usepackage{algorithm}
\usepackage{algorithmic}
\usepackage{array}
\usepackage{mdwmath}
\usepackage{mdwtab}
\usepackage{fancyhdr}

\begin{document}
\title{Distributed Control for Charging Multiple Electric Vehicles with Overload Limitation}
\author{Bo~Yang,~Jingwei~Li,~Qiaoni~Han,~Tian~He,~Cailian~Chen,~Xinping~Guan
\thanks{Copyright (c) 2015 IEEE. Personal use of this material is permitted. Permission from IEEE must be obtained for all other uses, including reprinting/republishing this material for advertising or promotional purposes, collecting new collected works for resale or redistribution to servers or lists, or reuse of any copyrighted component of this work in other works.}
\thanks{B. Yang, Q. Han, C. Chen, and X. Guan are with the Department of Automation, Shanghai Jiao Tong University, and Key Laboratory of System Control and Information Processing, Ministry of Education of China, Shanghai, 200240 P. R. China (e-mail: \{bo.yang, qiaoni, cailianchen, xpguan\}@sjtu.edu.cn).}
\thanks{J. Li was with the Department of Automation, Shanghai Jiao Tong University, Shanghai, 200240 P. R. China. She is currently with Centre for Translational Medicine (MD6), National University of Singapore, 14 Medical Drive, Singapore 117599 (email:  jingwei.li@u.nus.edu).} %
\thanks{T. He is with Department of Computer Science and Engineering, University of Minnesota, MN, USA. Email: tianhe@cs.umn.edu.}}%
\maketitle
\thispagestyle{fancy}
\fancyhead[L]{DOI: 10.1109/TPDS.2016.2533614}
\cfoot{}
\vspace{-0.5cm}
\begin{spacing}{2.0}
\begin{abstract}
Severe pollution induced by traditional fossil fuels arouses great attention on the usage of plug-in electric vehicles (PEVs) and renewable energy. However, large-scale penetration of PEVs combined with other kinds of appliances tends to cause excessive or even disastrous burden on the power grid, especially during peak hours. This paper focuses on the scheduling of PEVs charging process among different charging stations and each station can be supplied by both renewable energy generators and a distribution network. The distribution network also powers some uncontrollable loads. In order to minimize the on-grid energy cost with local renewable energy and non-ideal storage while avoiding the overload risk of the distribution network, an online algorithm consisting of scheduling the charging of PEVs and energy management of charging stations is developed based on Lyapunov optimization and Lagrange dual decomposition techniques. The algorithm can satisfy the random charging requests from PEVs with provable performance. Simulation results with real data demonstrate that the proposed algorithm can decrease the time-average cost of stations while avoiding overload in the distribution network in the presence of random uncontrollable loads.
\end{abstract}
\vspace{-0.5cm}
\begin{IEEEkeywords}
Electric vehicle, charging scheduling, renewable energy, smart grid, Lyapunov optmization.
\end{IEEEkeywords}
\vspace{-0.5cm}
\section{Introduction}
Nowadays pollution engendered by massively burning fossil fuel is serious. Hence more attention is paid on electric vehicles (EVs), to replace the traditional vehicles. The report by International Energy Agency (IEA) \cite{Trigg} shows that the number of global EV sales in 2012 reached 113,000, double of sales in 2011, and predicts that the expected number of EV on the road by 2020 will be up to 20 million. Therefore, massive plug-in electric vehicles (PEVs)/EVs will appear on the road in the near future, leading to huge amount of charging demands. This new kind of demand will surely be a challenge for the power grid, according to the survey conducted by Su \emph{et al.} \cite{SuEichi}. The safety of distribution network to support the huge demand from PEVs/EVs is essential in regulating the demand from vehicles. Parts of EVs' influences on the power grid have been studied in \cite{WuMohsenian-Rad}. Good simulation frameworks in \cite{Darabi} also show the interactions between plug-in hybrid electric vehicles (PHEVs) and the power grid.

Meanwhile, it is possible to integrate more renewable energy sources into the power network due to the two-way communication and regulation capability of smart grid. Studies in \cite{WuSun} give some inspiration for better utilization of distributed renewable energy generation. However, the uncertainty and intermittence of such renewable sources makes it difficult to use them stably and consecutively. A method to tackle this problem is to introduce storage devices, which can store renewable energy when it is abundant, and release the stored energy when renewable sources do not work. The objective of storage is similar with that in Rechargeable Sensor Networks \cite{HeChen}. Moreover, the storage device can be used to cut down the electricity cost by serving the load when electricity price is high. However, the time-diversity pattern of electricity price does not always match that of renewable energy generation, which depends on weather conditions mostly. Thus, how to minimize the non-renewable energy cost through optimal management of the storage device, renewable energy and on-grid energy comes to be a practical problem. As to PEV charging, there are already some real applications of renewable-powered charging stations, such as the wind powered stations in Barcelona, Spain \cite{ugei} and the charging stations with roof-top solar panels in Downtown Westport, Connecticut, USA \cite{Tushar}. Although they are not massively deployed, these examples demonstrate the technique feasibility for charging EVs/PEVs by renewable energy.

Prior works \cite{Rahbari-Asr}-\cite{Ghavami} have considered regulations of EV loads in a distribution network. Rahbari-Asr \emph{et al.} in \cite{Rahbari-Asr} consider to maximize the total utility of PEV/EV users, as well as avoiding line/node overload. The distributed charging regulation algorithms are designed to ensure certain system optimality and line constraints by means of utility maximization framework and game theory in \cite{Ardakanian} and \cite{Ghavami}, respectively. The scheme in \cite{Ardakanian} also considers the effects of residential loads on overloading the same distribution network that provides charging service to EVs. The similar effects are further discussed in \cite{Safdarian}-\cite{MaHu}. In general the above works regard EVs as time shiftable loads and regulate them to charge at different time slots in order to maintain power network safety and decrease electricity cost by utilizing the time diversity of electricity price. In other words, \cite{Ardakanian}-\cite{Angelis} and \cite{Sun1} are more suitable for slow charging scenario. As for fast charging, it usually lasts for less than one hour with large power consumption. Thus where to charge turns to be another practical problem, since it is important to avoid overload at individual place of a distribution network. Saovapakhiran \emph{et al.} propose to guide EVs and HEVs to the optimal charging station in order to maximize system utility in \cite{Saovapakhiran}. \cite{Wang} offers a coordinated charging strategy relying on the vehicle network communication to enhance the cooperation among mobile EVs. Compared with these papers that utilize conventional power grid to supply charging vehicles, Misra \emph{et al.} in \cite{Misra} consider to use pricing to attract EVs to charge at different microgrids, which can be powered by renewable energy sources. However, the intermittent nature of renewable energy is not considered in \cite{Misra}. Recently, we design an online algorithm to guide EV charging at different stations, which are supplied by both on-grid energy and time-varying renewable energy \cite{LiYang}.

As to dealing with the stochastic issue of renewable energy generation, Lyapunov optimization technique \cite{Georgiadis} has been applied to energy systems \cite{Neely1}-\cite{Chen} and wireless communications in smart grid \cite{Shah}. There are also other techniques that can be used to exploit the time-varying nature of renewable energy, such as Markov decision process in \cite{Zhang} and game theory in \cite{Maharjan}. However, the problem formulation in \cite{Zhang}-\cite{Maharjan} has to rely on the priori knowledge of the underlying stochastic process. This paper also utilizes Lyapunov optimization method to regulate PEVs to charge with on-grid power or renewable energy. Moreover, this paper further considers that each charging station is equipped with non-ideal energy storage with inefficient charging and discharging to store extra renewable energy. Thus, the framework of the algorithm designed for ideal storage in the aforementioned articles cannot be applied to this paper directly.

In this paper, a distribution network is considered, where there are multiple charging stations. Also many non-EV loads are connected to every transformer and transmission line in this distribution network. It is supposed that each charging station is equipped with a renewable energy generator and a non-ideal storage, which will absorb the abundant renewable energy. If the energy generated by renewable source is insufficient, charging stations are supposed to draw energy from the local storage or gain energy from the power grid. Hence this paper tries to minimize the time average non-renewable energy cost in a distribution network by coordinating the charging of PEVs and renewable energy management, while satisfying the safety of distribution networks by limiting the overload probability under a certain permitted bound. The contribution of this paper are threefold:

\begin{enumerate}[\indent 1.]
\item A stochastic optimization problem is formulated to minimize the time-average cost for serving PEVs by renewable energy with non-ideal storage and power networks that supply some uncontrollable loads at the same time.
\item An online algorithm consisting of guidance of PEV charging and energy management at stations is proposed based on Lyapunov optimization and Lagrange dual decomposition techniques. The algorithm can result in provable performance with non-ideal storage and ensure the safety of the distribution network.
\item Simulations with real data demonstrate that the asymptotic optimality and comparative waiting time can be achieved.
\end{enumerate}

The rest of this paper is organized as follows: system models are introduced in Section \uppercase\expandafter{\romannumeral2}; problem formulation is given in Section \uppercase\expandafter{\romannumeral3}; an online algorithm is developed in Section \uppercase\expandafter{\romannumeral4}; mathematical proofs of algorithm performance can be found in Section \uppercase\expandafter{\romannumeral5}; simulations are demonstrated in Section \uppercase\expandafter{\romannumeral6}; at last, conclusion of this paper is given in Section \uppercase\expandafter{\romannumeral7}. The key mathematical notations used in the paper are listed in Table 1.
\vspace{-0.5cm}
\begin{table}[width=7.5cm]
  \centering
  \caption{Mathematical notations}
  \label{tab:1}
  \footnotesize
  \begin{tabular}{l|l}
  \hline
  Notation    &Physical interpretation \\
  \hline
  \hline
  $l,\textbf{L}$         &Index and set of transformers and lines (nodes), $\textbf{L}\!=\!\{1,2,...,L\}$  \\
  $i,\textbf{I}$       &Index and set of charging stations, $\textbf{I}\!=\! \{1,2,3,...,N\}$ \\
  $j,\textbf{J}_i$       &Index and set of charging outlets belonging to the $i$-th charging station, $\textbf{J}_i\!=\!\{1,2,...,J_i\}$  \\
  $k,\textbf{K}$   &Index and set of entry points, $\textbf{K}\!=\!\{1,2,...,K\}$ \\
  $[tT, (t+1)T)$          &$t$-th interval with $T$ being the length of a slot   \\
  $P_l$      &Capacity of the $l$-th transformer/line \\
  $X_{li}$     &Indicating variable for the relationship between the $l$-th node and the $i$-th charing station  \\
  $N_l(t)$	        &Amount of uncontrollable loads connected to the $l$-th node in time slot $t$ \\
  $U_i(t)$	        &Production rate of a renewable energy generator for thr $i$-th station \\
  $B_i(t)$          &Energy level of the storage device for the $i$-th charging station \\
  $R_i(t)$       &Charging rate renewable energy charged into the $i$-th charging station's battery \\
  $r_{ij}(t)$      &Charging rate of the charging outlet $j$ in the $i$-th charging station \\
  $D_i^d(t)$       &Energy flow from the grid to charging outlets directly \\
  $\epsilon$       &Maximum probability of the overload chance \\
  $c(t)$        &Power price in time slot $t$ in the system \\
  $\sigma_{N_l}$   &Standard deviation of $N_l(t)$  \\
  $e_k(t)$      &Amount of energy demand at the $k$-th entry point in time slot $t$ \\
  $\eta_k(t)$      &Whether there is a PEV carrying an amount of energy demand of $e_k(t)$ in time slot $t$\\
  $E_k(t)$    &Energy demand arrives at the $k$-th entry point in time slot $t$ \\
  $p_k$    &Probability of PEVs sending request to the $k$-th entry point \\
  $W_{k,ij}(t)$   &Indicating variable for connections between $k$-th entry point and $i$-th charging station's $j$-th outlet in time slot $t$\\
  $Q_{ij}(t)$  &Demand queue of the $j$-th outlet of the $i$-th charging station at the beginning of time slot $t$ \\
  $\eta^+,\eta^-$  &Charging efficiency, discharging efficiency \\
  $H_i(t)$     &Shift version of $B_i(t)$ \\
  $\overrightarrow{\theta(t)}$   &State of system queues at the beginning of time slot $t$ \\
  $\mathcal L(\overrightarrow{\theta(t)})$  &Lyapunov function \\
  $\triangle(\overrightarrow{\theta(t)})$   &Lyapunov drift\\
  \hline
  \end{tabular}
\end{table}
\vspace{-0.8cm}
\section{System Models}
This section introduces the system model of a distribution network, which includes several charging stations, uncontrollable non-PEV loads, and random request of PEVs.
\vspace{-0.5cm}
\subsection{Distribution Network and Charging Stations}
In this paper, a distribution network (Fig. 1), composed of multiple transformers and lines, is considered. The set of transformers and lines is $\textbf{L}\!=\!\{1,2,...,L\}$. Every transformer and transmission line has its power capacity, determined by the type of conductor, the highest temperature permitted, and other physical and environmental factors, and the capacity of the $l$-th ($l\!\in\!\textbf{L}$) transformer/line is denoted by $P_l$. The considered system is modeled as a slotted one and it is referred to the $t$-th interval $\![tT, (t+1)T\!)$, where $T$ is the length of a slot. In the following $T$ is normalized to unit for easier presentation.
\begin{figure}
\centering
\includegraphics[width=90mm]{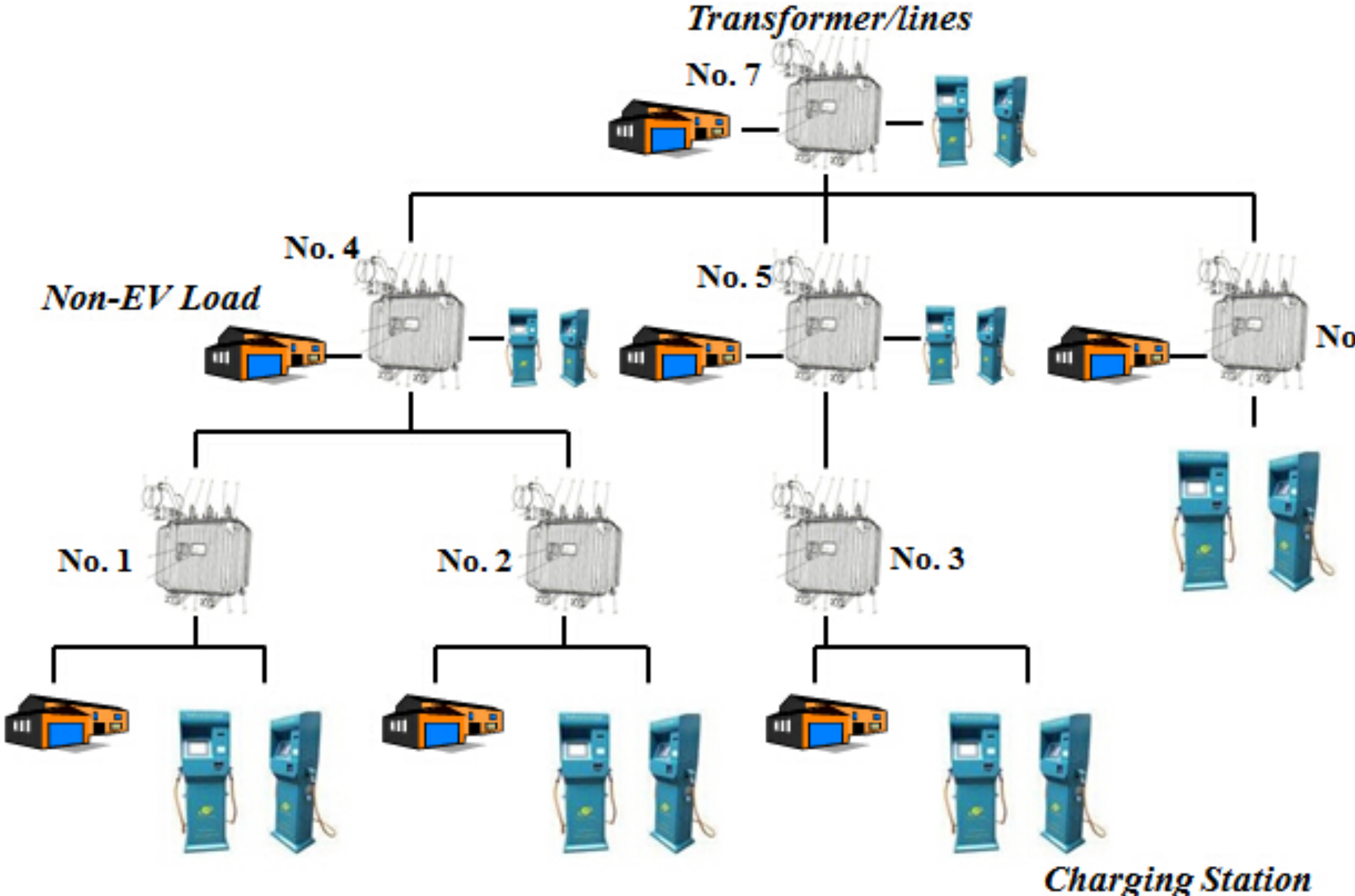}
\caption{Distribution network.}\label{fig_1}
\end{figure}

Let $i\!\in\!\textbf{I}\!=\!\{1,2,\cdots,N\}$ denote the $i$-th charging station and $j\!\in\!\textbf{J}_i\!=\!\{1,2,...,J_i\}$ denote the $j$-th charging outlet in the $i$-th station, where $\textbf{J}_i$ is the set of charging outlets belonging to the $i$-th charging station. In this distribution network, different charging stations are connected to different transformers/lines, represented by different nodes later for consistency. Therefore, the matrix $\textbf{X}_{(L\times I)}$ is used to indicate the relationship between different nodes and charing stations, where $X_{li}\!=\!1$ means that the $i$-th station is downstream of the $l$-th node, otherwise $X_{li}\!=\!0$. It should be noted that the network constraint is not explicitly modelled here as in generic AC networks. We only use the matrix $\mathbf{X}_{\!\left(L\times I\!\right)}$ to reflect the DC power flow relationship between different nodes and charging stations to present the problem clearly. Meanwhile, other non-EV loads, such as factories, households, and shopping malls, are also fed by this distribution network. The amount of uncontrollable loads connected to the $l$-th node is $N_l(t)$ in time slot $t$.

To alleviate the burden imposed by charging stations on the distribution networks, every charging station is equipped with a renewable energy generator, whose production rate is $U_i(t)\!\in\![0,U_{max}]$ for the $i$-th station, and the cost to use its own renewable energy is neglected in this scenario. We do not consider the initial investment cost for deploying renewable energy generators, since it is a constant. Note that it is easy to incorporate the storage degradation cost due to charging/discharging by the model in \cite{Sun2}. Since the renewable energy production is intermittent and unpredictable, a storage device, for example a battery, is equipped for every charging station. It can also cut down the electricity cost since a battery as a buffer can store energy when electricity price is low and release energy for EVs' charging when price is high, to reduce the cost. The energy level of the storage device for the $i$-th charging station is $B_i(t)$, bounded by $[0,B_{i,max}]$, and $B_{i,max}\!\in\! (0,B_{max}]$. In the following, we refer the battery as the storage device in each station except for other specification. Given that the SOC of each battery is upper bounded by the battery's capacity, not all amount of renewable energy $U_i(t)$ in time slot $t$ is charged into it.

For every charging outlet $j$ in the $i$-th charging station, the charging rate $r_{ij}\!(t\!)$ is also constrained by $r_{ij}(t)\!\in\![0,r_{i,max}]$ and $r_{i,max}\in(0,r_{max}]$. The PEV charging load is firstly supplied by harvested energy from renewable energy generator. If there is still redundant energy, it is then charged into the battery\footnote{If there is still superfluous energy after charging the storage, the renewable generator will self-regulate its generation rate. For example, the output power of a solar panel can be changed by maximum power point tracking method. If the renewable energy generator is connected to the power grid, the redundant energy will be input to the grid. However, in this paper we assume that the local renewable energy generator is run in off-grid mode.}. Let $R_{i}\!\left(t\!\right)$ be the current charging rate to the battery and thus $R_{i}\!\left( t\!\right)\!\in\!\left[0,\min\!\left\{R_{i}^{\max},\!\left(U_{i}\!\left(t\!\right)\!-\!\sum\limits_{j=1}^{J_{i}}\!r_{ij}\!\left( t\!\right)\!\right)^{+}\!\right\}\!\right]$, where $\!\left( x\!\right) ^{+}\!\hat{=}\!\max\!\left\{ 0,x\!\right\}$. Otherwise an amount of $\left(\sum\limits_{j=1}^{J_{i}}r_{ij}\left(t\right)-U_{i}\left(t\right)\right)^{+}$ load will be served either by the battery or by the distribution networks depending on the energy availability in battery and usage cost of power grid. Let $D_{i}^{d}\left(t\right)\!\in\!\left[0,D^{d,\max}\!\right] $ denote the amount of energy withdrawn from the distribution network and $\!\left(\sum\limits_{j=1}^{J_{i}}r_{ij}\!\left(t\!\right)\!-\!U_{i}\!\left(t\!\right)\!\right)^{+}\!\!-\!\!D_{i}^{d}\!\left(t\!\right)$ the amount of energy supplied by the battery. To ensure that the distribution network can sustain the charging load independently, we assume $D^{d,\max}\!>\!r_{i,\max }$. Thus the total loads imposed on the $l$-th node is $\sum_{i:X_{li}=1}\!D_i^d(t)\!+\!N_l(t)$.

If no overload is permitted, $\sum\nolimits_{i:X_{li}=1}\!D_i^d(t)\!+\!N_l(t)\!\leq\!P_l$ should stand for every time slot $t$. Since the uncontrollable loads are time varying and each node always allows slight overload for the safety reason \cite{Ardakanian}, the above inequality can be slacked to the ``chance constraint'':
\vspace{-0.5cm}
\begin{equation}\label{eq_1}
\begin{small}
Pr[\sum\nolimits_{i:X_{li}=1}\!D_i^d(t)\!+\!N_l(t)\!>\!P_l]\!\leq\!\epsilon,
\end{small}
\end{equation}where $\epsilon$ is chosen to ensure that the node is safe even with a little overload.

At the end of this subsection, let $c(t)$ denote the power price in time slot $t$ in the system, and it is also bounded by $[C_{min},C_{max}]$.
\vspace{-0.5cm}
\subsection{PEVs' Request and Guidance}
In the area supplied by this distribution network, there are numerous entry points for PEVs to join, which can be considered as fixed data uploading workstations on the roadside, and the number of such entry points is $K$. Here we just use the entry point to describe the parallel process of multiple PEVs' request. PEVs randomly send charging requests to an idle entry point, and the variable $\xi_k(t)$ denotes whether there is a PEV carrying an amount of energy demand of $e_k(t)$, sending request to the $k$-th ($k\!\in\!\textbf{K}\!=\!\{1,2,\cdots,K\}$) entry point. If there is a PEV uploading request to the $k$-th entry point, let $\xi_k(t)\!=\!1$, otherwise $\xi_k(t)\!=\!0$. Therefore the energy demand arrives at the $k$-th entry point is $E_k(t)\!=\!e_k(t)\xi_k(t)\!\leq\!E_{max}\!<\!\infty$. Notably, the probability of PEVs sending request to the $k$-th entry point is denoted by $p_k$, which means $Pr[\xi_k(t)=1]=p_k$.

The proposed directing policy in Section \uppercase\expandafter{\romannumeral4} decides which charging outlet the PEVs arriving at each entry point should be guided to. Hence the binary variable $W_{k,ij}(t)\!\in\!\{0,1\}$ is used to denote the connections between the $k$-th entry point and the $j$-th outlet of the $i$-th charging station, where $W_{k,ij}(t)\!=\!1$ means the PEVs sending request to the $k$-th entry point shall be guided to the $j$-th outlet in the $i$-th charging station. It is reasonable that each charging outlet can only serve one entry point and only one outlet can be connected to each entry point simultaneously. Therefore, $W_{k,ij}(t)$ satisfies
\vspace{-0.5cm}
\begin{equation}\label{eq_2}
\begin{small}
\sum\nolimits_{i=1}^{I}\sum\nolimits_{j=1}^{J_i}\!W_{k,ij}(t)\!\leq\!1 ~~(\forall k\!\in\!\textbf{K}),
\end{small}
\end{equation}
\vspace{-1.0cm}
\begin{equation}\label{eq_3}
\begin{small}
\sum\nolimits_{k=1}^K\!W_{k,ij}(t)\!\leq\!1 ~~(\forall i\!\in\!\textbf{I}, j\!\in\!\textbf{J}_i).
\end{small}
\end{equation}

The dynamic operation of this system is stated here. At the beginning of every time slot, the central controller collects the energy demand at each entry point and some information of each charging station before broadcasting the guidance information. At the same time, every node computes the overload indication according to current PEV loads and uncontrollable loads imposed on it. Finally, each charging station computes other control variables $D_i^d(t)$ and $R_{i}(t)$ in a distributed way based on the local information and the overload indication.
\vspace{-0.5cm}
\subsection{Battery Dynamics and Demand Queue}
For each charging station, its battery's energy state satisfies the following iterating equation:
\vspace{-0.3cm}
\begin{equation}\label{eq_4}\begin{small}
B_{i}\!\left(t+1\!\right)\!=\!B_{i}(t)\!-\!\eta^{-}\!\left(\!\left(\!\sum\limits_{j=1}^{J_{i}}\!r_{ij}(t)\!-\!U_{i}
(t)\!\right)^{+}\!-\!D_{i}^{d}(t)\!\right)\!+\!\eta^{+}\!R_{i}(t),
\end{small}\end{equation}where $0\!<\!\eta^{+}\!\leq\!1$ and $\eta ^{-}\!\geq\!1$ represent the charging and discharging efficiency respectively, and $\eta ^{-}\left(\!\left(\sum\limits_{j=1}^{J_{i}}r_{ij}(t)\!-\!U_{i}(t)\!\right)^{+}\!-\!D_{i}^{d}(t)\!\right)$ denotes the amount of energy that the $i$-th charging station extracts from its battery. The following inequality stands for every battery:
\vspace{-0.3cm}\begin{equation}\label{eq_5}\begin{small}
0\!\leq\!\eta ^{-}\!\left(\!\left(\!\sum\limits_{j=1}^{J_{i}}r_{ij}\left( t\right)\!-\!U_{i}\left( t\right)\!\right) ^{+}\!-\!D_{i}^{d}\left( t\right) \right)\!\leq\!B_i(t)\!\leq\!B_{i,max},
\end{small}\end{equation}which means charging outlets in the $i$-th charging station cannot obtain amounts of energy from the battery more than the battery's residual energy and the battery also has its capacity constraint. The algorithm given in Section \uppercase\expandafter{\romannumeral4} can guarantee it perfectly.

Since PEVs' charging request is random, the following virtual demand queues are introduced to measure the relationship between charging request and energy supply. The stability of demand queues, defined in Section \uppercase\expandafter{\romannumeral3}, can be reached by regulating charging rate at every outlet and directing policy in Section \uppercase\expandafter{\romannumeral4} to realize the supply-demand balance. The stability of demand queue means that the time-average charging requests can be satisfied. For every outlet, the demand queue is denoted by:
\vspace{-0.3cm}
\begin{equation}\label{eq_6}
\begin{small}
Q_{ij}(t+1)=[Q_{ij}(t)-r_{ij}(t)+E_k(t)W_{k,ij}(t)1_{\{Q_{ij}(t)=0\}}]^+,
\end{small}
\end{equation}
where $[a]^+$ stands for $\max\{a,0\}$. The expression $1_{\{Q_{ij}(t)=0\}}$ here denotes that only if the demand queue is empty at the beginning of time slot $t$, can any new PEV's demand be directed to this charging outlet.
\vspace{-0.3cm}
\section{Problem Formulation}
The goal of this paper is to minimize the time-average expected cost from the power grid with constraints of battery and stable demand queues, as well as the safety of the distribution networks. The problem can be formulated as a stochastic optimization problem.
\vspace{-0.3cm}\begin{align}
\textbf{P1:}~~&\min~\lim _{\tau\rightarrow\infty}\frac{1}{\tau}\sum_{t=0}^{\tau-1}\sum_{i=1}^I\textbf{E}\{c(t)D_i^d(t)\} \\
&\mbox{s.t.}\,\,\text{stability of}~~ Q_{ij}(t)~~(\forall i\in \textbf{I},\forall j\in \textbf{J}_i), (1), (2), (3), (4), (5), (6) \notag
\end{align}where the expectation above is with respect to the potential randomness of the control policy. The stability of these queues are defined in Definition 1.
\begin{Definition}\label{def_1}
A discrete queue evolving like $Q(t+1)=[Q(t)-S(t)]^+ +D(t)$ is \emph{strongly stable} if:
\vspace{-0.3cm}
\[\lim_{t\rightarrow \infty}\sup \frac{1}{t}\sum_{\tau =0}^{t-1}\textbf{E}\{Q(\tau)\}<\infty.\]
\end{Definition}

If demand queues including the battery dynamics iteration (4) are \emph{strongly stable}, the system is stable. According to Neely's monograph \cite{Neely2}, a \emph{strongly stable} queue must be \emph{rate stable} for finite service $S(t)$ and stochastic arrival $D(t)$. Furthermore, if and only if $d\leq s$ holds, $Q(t)$ is \emph{rate stable}, where $d=\lim_{t\rightarrow \infty }\frac{1}{t}\sum_{\tau =0}^{t-1}\textbf{E}\{D(\tau)\}$ and $s=\lim_{t\rightarrow \infty }\frac{1}{t}\sum_{\tau =0}^{t-1}\textbf{E}\{S(\tau)\}$.

There is time coupling difficulty in solving P1 due to (4), that is to say the current charging and discharging decisions will affect the energy level in battery in the future. To tackle this difficulty, a relaxed problem will be studied:
\vspace{-0.5cm}\begin{align*}
\widetilde\textbf{P1}\textbf{:}~&\min~~
\lim_{\tau\rightarrow\infty}\frac{1}{\tau}\sum_{t=0}^{\tau-1}\!\sum_{i=1}^I\!\textbf{E}\{c(t)D_i^d(t)\} \\
\mbox{s.t.}\,\,&\text{stability of}~Q_{ij}(t)~(\forall i\in \textbf{I},\forall j\in \textbf{J}_i), (1), (2), (3), (6),\\
&\lim_{\tau \rightarrow \infty }\frac{1}{\tau }\sum\limits_{t=0}^{\tau -1}\!\mathbf{E}\!\left\{ \eta ^{-}\!\left( \!\left(\!\sum\limits_{j=1}^{J_{i}}r_{ij}\left( t\right)\!-\!U_{i}\left( t\right)\!\right)^{+}\!-\!D_{i}^{d}\left( t\right)\!\right)\!\right\}\!\geq \!\lim_{\tau \rightarrow\infty}
\frac{1}{\tau}\!\sum\limits_{t=0}^{\tau-1}\!\mathbf{E}\!\left\{ \eta^{+}R_{i}\left(t\right)\!\right\}.\end{align*}

Denote the optimal value of \textbf{P1} as $P^*$ and $\widetilde \textbf{P1}$ as $P_{re}^*$. Notice that any feasible solution of \textbf{P1} is also a feasible solution of $\widetilde \textbf{P1}$ since $\widetilde \textbf{P1}$ is less constrained than \textbf{P1}. Therefore, $P_{re}^*\leq P^*$. We will solve $\widetilde \textbf{P1}$ first. Later we will prove the performance gap between the proposed algorithm and the optimal solution to \textbf{P1}. More interestingly, the proposed algorithm for $\widetilde\textbf{P1}$ also satisfies the constraints of (\ref{eq_5}).
\vspace{-0.3cm}
\section{Problem Decomposition and Solution}
One challenge in solving $\widetilde \textbf{P1}$ is due to the fact that the supply, demand and price are all stochastic processes. If all statistical information is known, it can be solved offline. However, the solution to this problem should be achieved by real-time implementations.
Firstly, some transformation of constraint (1) shall be done to write it into a deterministic form.
\vspace{-0.5cm}
\subsection{Transformation of Chance Constraint}
Assume the distribution of uncontrollable load $N_l(t)$ is symmetric about its mean. Then there is
\vspace{-0.3cm}
\begin{equation}\label{eq_8}\begin{split}
Pr[N_l(t)\!\!\geq\!\! P_l\!\!-\!\!\!\!\!\!\displaystyle \sum_{i:X_{li}=1}\!\!\!\!D_i^d(t)]\!\!
&=\!\!\frac{1}{2}Pr[|N_l(t)\!\!-\!\!\textbf{E}\{\!N_l(t)\!\}|\!\!\geq\!\! P_l\!\!-\!\!\!\!\!\sum_{i:X_{li}=1}\!\!\!\!D_i^d(t)\!\!-\!\!\textbf{E}\{N_l(t)\}]\\
&\leq \!\frac{1}{2} \frac{\sigma_{N_l}^2}{[P_l-\!\!\displaystyle\sum_{i:X_{li}=1}D_i^d(t)\!-\!\textbf{E}\{N_l(t)\}]^2}\leq \epsilon,
\end{split}\end{equation}where $\sigma_{N_l}$ is the standard deviation of $N_l(t)$. As in \cite{Tarasak}, Chebyshev's inequality is given below:
\vspace{-0.3cm}
\begin{equation*}
\begin{small}
Pr[|x|\geq a]\leq{\sigma_x^2}/{a^2}
\end{small}
\end{equation*}
for a zero-mean random variable $x$, and $\sigma_x$ is the standard deviation of $x$. Therefore, the constraint in (\ref{eq_8}) is transferred to
\vspace{-0.5cm}
\begin{equation}\label{eq_9}
\begin{small}
P_l\!-\!\displaystyle\sum_{i:X_{li}=1}\!D_i^d(t)\!-\!\textbf{E}\!\{N_l(t)\!\}\!\geq\!\frac{\sigma_{N_l}}{\sqrt{2\epsilon}}.
\end{small}
\end{equation}
\vspace{-1.5cm}
\subsection{Real-time Decomposition}
In order to develop algorithm to regulate the SOC of each battery to satisfy (\ref{eq_5}), the shifted version of $B_i(t)$, $H_i(t)\!=\!B_i(t)\!-\!T_{max}\!-\!\eta^-\!J_ir_{i,max}$ is given, where $T_{max}$ is a constant which will be defined in Section \uppercase\expandafter{\romannumeral5}. Like $B_i(t)$, $H_i(t)$ also satisfies
\vspace{-0.5cm}
\begin{equation}\label{eq_10}
\begin{small}
H_{i}\!\left( t+1\!\right)\!=\!H_{i}\left( t\right)\!-\!\eta ^{-}\!\left(\!\left(\!\sum\limits_{j=1}^{J_{i}}r_{ij}\left( t\right)\!-\!U_{i}\left( t\right)\!\right)^{+}\!-\!D_{i}^{d}\left( t\right)\!\right)\!+\!\eta ^{+}R_{i}\left( t\right).
\end{small}\end{equation}
~~The technique for shifted queue is also utilized in \cite{Guo}. Then the state of system queues is described as $\overrightarrow{\theta(t)}=(Q_{11}(t),...,Q_{IJ_I}(t),H_1(t),...,H_I(t))$.

Define Lyapunov function as:
\vspace{-0.5cm}
\begin{equation}\label{eq_11}
\begin{small}
\mathcal{L}(\overrightarrow{\theta(t)})=\frac{1}{2}\sum_{i=1}^I\sum_{j=1}^{J_i}Q_{ij}^2(t)+\frac{1}{2}\sum_{i=1}^IH_i^2(t).
\end{small}
\end{equation}

The Lyapunov drift is given by
\vspace{-0.5cm}
\begin{equation}\label{eq_12}
\begin{small}
\triangle(\overrightarrow{\theta(t)})=\textbf{E}\{\mathcal{L}(\overrightarrow{\theta(t+1)})-\mathcal{L}(\overrightarrow{\theta(t)})| \overrightarrow{\theta(t)}\}.
\end{small}
\end{equation}

From (\ref{eq_6}) and (\ref{eq_10}), it is easy to see that
\vspace{-0.5cm}
\begin{equation} \label{eq_13}
\begin{small}
 \frac{1}{2}(Q_{ij}^2(t+1)\!-\!Q_{ij}^2(t))\!\leq\!\alpha_{\max}\!-\!Q_{ij}(t)r_{ij}(t)
 \!-\!r_{ij}(t)\!\sum\nolimits_{k=1}^{K}\!E_k(t)W_{k,ij}(t)1_{\{Q_{ij}(t)=0\}},
 \end{small}\end{equation}
\vspace{-1cm}
\begin{equation} \label{eq_14}
\begin{small}
\frac{1}{2}(H_{i}(t+1)\!-\!H_{i}^2(t))\!\leq\!\beta_{\max}\!-\!H_{i}(t)\!\left(\eta^{-}\!\left(\!\left(\!\sum\limits_{j=1}^{J_{i}}r_{ij}(t)\!-\!U_{i}( t)\right)^{+}\!-\!D_{i}^{d}(t)\!\right)\!+\!\eta ^{+}R_{i}(t)\right),
\end{small}
\end{equation}
where $\alpha_{\max}\!\!=\!\!\frac{1}{2}r_{max}^2\!\! +\!\! \frac{1}{2}E_{max}^2\!\! <\!\! \infty$ and $\beta _{\max }=\frac{1}{2}\max \left\{ \left( \eta ^{+}U_{\max }\right)^{2},\max_{i}\left\{ J_{i}^{2}r_{i,\max }^{2}\right\} \right\}\!\!<\!\!\infty$ are finite constants.

Substitute (\ref{eq_13}) and (\ref{eq_14}) into (\ref{eq_12}). We consider the following drift plus penalty expression as in \cite{Jin}. The Lyapunov drift is used to measure stability and the penalty item accounts for system performance. The tradeoff between system stability and optimal performance is achieved by the parameter $V$.
\vspace{-0.3cm}
\begin{equation} \label{eq_15}
\begin{split}
&\triangle(\overrightarrow{\theta(t)})\!+\!V\textbf{E}\{c(t)\sum_{i=1}^{I}D_i^d(t)|\overrightarrow{\theta(t)}\}
\!\leq\!\delta_{\max}\!-\!\sum\limits_{i=1}^{I}\!\mathbf{E}\!\left\{\sum\limits_{j=1}^{J_{i}}r_{ij}\left( t\right)\!\left( Q_{ij}\left( t\right)
\!+\!H_{i}\left( t\right)\eta ^{-}\!\right)\!\cdot\!1_{\left\{ Q_{ij}\left(
t\right)>0\right\} }\left\vert \overrightarrow{\theta \left( t\right) }%
\right.\!\right\}\\
&\!-\!\sum\limits_{i=1}^{I}\mathbf{E}\left\{
\sum\limits_{j=1}^{J_{i}}r_{ij}\left( t\right) \left(
\sum\limits_{k=1}^{K}E_{k}\left( t\right) W_{k,ij}\left( t\right)
\!+\!H_{i}\left( t\right) \eta ^{-}\right) \cdot 1_{\left\{ Q_{ij}\left(
t\right) =0 \sum\limits_{j=1}^{J_{i}}r_{ij}\left( t\right) \geq
U_{i}\left( t\right) \right\} }\left\vert \overrightarrow{\theta \left(
t\right) }\right. \right\}\\
&\!+\!\sum\limits_{i=1}^{I}\!\mathbf{E}\!\left\{ H_{i}\left( t\right) R_{i}\left(
t\right) \eta ^{+}\!\cdot\!1_{\left\{ \sum\limits_{j=1}^{J_{i}}r_{ij}\left(
t\right)\!<U_{i}\left( t\right)\!\right\} }\!\left\vert\! \overrightarrow{\theta
\left( t\right) }\right. \right\}\!+\!\sum\limits_{i=1}^{I}\mathbf{E}%
\!\left\{ H_{i}\left( t\right) \eta ^{-}U_{i}\left( t\right)\!\cdot\!1_{\left\{
\sum\limits_{j=1}^{J_{i}}r_{ij}\left( t\right) \geq U_{i}\left( t\right)
\right\} }\left\vert \overrightarrow{\theta \left( t\right) }\right.
\right\}\\
&+\sum\limits_{i=1}^{I}\mathbf{E}\left\{
D_{i}^{d}\left( t\right) \left( H_{i}\left( t\right) \eta ^{-}\!+\!Vc\left(
t\right) \right) \cdot 1_{\left\{ \sum\limits_{j=1}^{J_{i}}r_{ij}\left(
t\right) \geq U_{i}\left( t\right) \right\} }\left\vert \overrightarrow{%
\theta \left( t\right) }\right. \right\},
\end{split}\end{equation}where $\delta_{max} \!\!= \!\!\sum_{i=1}^{I}J_i\alpha_{max}\!\! +\!\! I\beta_{max}\!\! <\!\! \infty$, and $V$ is an important parameter in regulating the trade-off between system cost and storage capacity. By adopting the framework of Lyapunov optimization, two sub-problems only relying on current information are derived by minimizing the right hand side of (\ref{eq_15}) approximately subject to the constraints in $\widetilde\textbf{P1}$. Later it is proven that the proposed algorithm can satisfy the constraints of (\ref{eq_5}). Notice that (\ref{eq_1}) has been replaced by (\ref{eq_9}).
\vspace{-0.5cm}
\begin{align}\label{eq_16}
&\textbf{P2:}~~\text{If}~~Q_{ij}(t)\!>\!0 \notag\\
&\sum\limits_{i=1}^{I}\!\left\{\!\sum\limits_{j=1}^{J_{i}}r_{ij}\left(
t\right) \left( Q_{ij}\left( t\right)\!+\!H_{i}\left( t\right) \eta ^{-}\right)
\!\cdot\!1_{\left\{ Q_{ij}\left( t\right)>0\right\}}\!+\!\sum\limits_{j=1}^{J_{i}}r_{ij}\left( t\right)\!\left(
\!\sum\limits_{k=1}^{K}\!E_{k}\left( t\right) W_{k,ij}\left( t\right)
\!+\!H_{i}\left( t\right) \eta ^{-}\right)\!\cdot\!1_{\left\{ Q_{ij}\left(
t\right) =0\right\}}\right.\notag\\
&\left.\!-\!D_{i}^{d}\left( t\right) \left( H_{i}\left( t\right)
\eta ^{-}\!+\!Vc\left( t\right) \right)\!-\!H_{i}\left( t\right) R_{i}\left(
t\right) \eta ^{+}\right\} \\
&\mbox{s.t.}\,\,(9), r_{ij}\left( t\right)\!\in\!\left[ 0,r_{i,\max }\right], R_{i}\left(
t\right)\!\in\!\left[ 0,\min\!\left\{ R_{i}^{\max },\!\left( U_{i}\left( t\right)
\!-\!\sum\limits_{j=1}^{J_{i}}\!r_{ij}\left( t\right)\!\right) ^{+}\!\right\}\!\right], \forall i\!\in\!\mathbf{I}.\notag
\end{align}

Another subproblem can be extracted from \textbf{P2}, to decide the connections between PEV requests at different entry points and charging outlets, for $\forall Q_{ij}(t)=0$, shown as below:
\vspace{-0.5cm}
\begin{align}\label{eq_17}
\textbf{P3:}~~&\max_{W_{k,ij}(t)}~~\sum\limits_{k=1}^{K}E_{k}\left( t\right) W_{k,ij}\left( t\right)
+H_{i}\left( t\right) \eta ^{-}\\
&\mbox{s.t.}\,\,(2), (3) \notag
\end{align}which means that the algorithm should meet the largest amount of demand with the empty charging outlet whose battery's residual energy is the most.

Notice that there is variable coupling due to the constraint (\ref{eq_9}) in \textbf{P2}. To decouple this constraint, Lagrange dual method will be employed in the next subsection.
\vspace{-0.5cm}
\subsection{Decoupling of Primal Problems}
Considering Lagrangian relaxation for \textbf{P2}, two functions are defined in the case of $Q_{ij}(t)>0$ and $Q_{ij}(t)=0$ respectively:
\vspace{-0.3cm}
\[
\begin{split}
&F_t^1(\overrightarrow{r}\!,\!\overrightarrow{D^d}\!,\!\overrightarrow{R}\!,\!\overrightarrow{\lambda})\!\!=\!\!\sum_{i=1}^I\{\sum_{j=1}^{J_i}\!r_{ij}(t)[Q_{ij}(t)\!\!+\!\!H_i(t)]\!-\!\!
D_i^d(t)(H_i(t)\eta^-\!\!+\!\!Vc(t))\!\!-\!\!H_i(t)R_{i}(t)\eta^+\}\!\\&
+\!\!\sum_{l=1}^L\lambda_l(P_l\!-\!\!\!\!\!\!\sum_{i:X_{li}=1}\!\!\!\!D_i^d(t)\!\!-\!\!\textbf{E}\{N_l(t)\}\!\!-\!\!\frac{\sigma_{N_l}}{\sqrt{2\epsilon}}),
\end{split}
\]
\vspace{-0.5cm}
\[
\begin{split}
&F_t^2(\!\overrightarrow{r}\!,\!\overrightarrow{D^d}\!,\!\overrightarrow{R}\!,\!\overrightarrow{\lambda})\!\!=\!\!\sum_{i=1}^I\!\{\sum_{j=1}^{J_i}\!r_{ij}(t)[\sum_{k=1}^K E_k(t)W_{k,ij}(t)\!\!+\!\!H_i(t)\eta^-]\!\!\\&
\!\!-\!\!(D_i^d(t)\!\!+\!\!D_i^b(t))(H_i(t)\eta^-\!\!+\!\!Vc(t))\!\!-\!\!H_i(t)R_{i}(t)\eta^+\}\!\!+\!\!
\sum_{l=1}^L\lambda_l(P_l\!\!-\!\!\sum_{i:X_{li}=1}D_i^d(t)\!\!-\!\!\textbf{E}\{N_l(t)\}\!\!-\!\!\frac{\sigma_{N_l}}{\sqrt{2\epsilon}}),
\end{split}
\]where $\overrightarrow{r}=(r_{11}(t),...,r_{IJ_i}(t))$, $\overrightarrow{D^d}=(D_1^d(t),...,D_I^d(t))$, $\overrightarrow{R}=(R_1(t),...,R_I(t))$, and $\overrightarrow{\lambda}=(\lambda_1,...,\lambda_L)$.

The dual functions $\Gamma_t^1(\overrightarrow{\lambda})$ and $\Gamma_t^2(\overrightarrow{\lambda})$ are defined as the partial maximum of $F_t^1(\overrightarrow{r},\overrightarrow{D^d},\overrightarrow{R},\overrightarrow{\lambda})$ and $F_t^2(\overrightarrow{r},\overrightarrow{D^d},\overrightarrow{R},\overrightarrow{\lambda})$ with respect to $\overrightarrow{r}$, $\overrightarrow{D^d}$, $\overrightarrow{R}$:
\vspace{-0.5cm}
\[
\Gamma_t^1(\overrightarrow{\lambda})=\max_{\overrightarrow{r},\overrightarrow{D^d},\overrightarrow{R}} F_t^1(\overrightarrow{r},\overrightarrow{D^d},\overrightarrow{R},\overrightarrow{\lambda}),~\Gamma_t^2(\overrightarrow{\lambda})=\max_{\overrightarrow{r},\overrightarrow{D^d},\overrightarrow{R}} F_t^2(\overrightarrow{r},\overrightarrow{D^d},\overrightarrow{R},\overrightarrow{\lambda})
\]
\vspace{-1.5cm}
\begin{tabbing}
\hspace*{20bp}
s.t.~~\bigskip $D_{i}^{d}\left( t\right)\!\!\in\!\!\left[ 0,\left(
\sum\limits_{j=1}^{J_{i}}r_{ij}\left( t\right)\!-\!U_{i}\left( t\right) \right)
^{+}\!\right], R_{i}\left( t\right)\!\!\in\!\!\left[0,\min\!\left\{ R_{i}^{\max },\!\left(
U_{i}\left( t\right)\!-\!\sum\limits_{j=1}^{J_{i}}r_{ij}\left( t\right) \right)^{+}\right\}\!\right], (\forall i\!\in\!\textbf{I},j\!\in\!\textbf{J}_i)$
\end{tabbing}

Now scrutinizing the dual functions above, it can be decomposed into subproblems for each charging outlet, renewable energy generator, and charging station. The following problems \textbf{P4(a)} and \textbf{P4(b)} are for every charging outlet to decide their charging rates. Problem \textbf{P5} is for station's private renewable energy generator to regulate the charging rate $R_{i}(t)$ to the battery. And problem \textbf{P6} is for every charging station to determine the energy they consume from the grid.

If $Q_{ij}(t)\!>\!0$,
\begin{equation}\textbf{P4(a)}:~~\max_{0\leq r_{ij}(t)\leq r_{i,max}}~~~r_{ij}(t)[Q_{ij}(t)\!+\!H_i(t)\eta^-] \label{eq_18}\end{equation}
\vspace{-0.5cm}

\vspace{-0.5cm}If $Q_{ij}(t)\!=\!0$,
\begin{equation}\textbf{P4(b)}:~~\max_{0\leq r_{ij}(t)\!\leq\!r_{i,max}}~~r_{ij}(t)[\sum_{k=1}^K\!E_k(t)W_{k,ij}(t)\!+\!H_i(t)\eta^-] \label{eq_19}\end{equation}$\sum_{k=1}^K E_k(t)W_{k,ij}(t)+H_i(t)\eta^-$ in \textbf{P4(b)} will be solved by directing policy in the next subsection, which is the solution of \textbf{P3}.
\vspace{-0.5cm}
\begin{align}
\textbf{P5:}~~&\min_{R_{i}\left( t\right) }~~\sum_{i=1}^{I}\!H_{i}\left( t\right)
R_{i}\left( t\right) \eta ^{+} \\ \label{eq_20}
&\mbox{s.t.}~~0\!\leq\!R_{i}\left( t\right)\!\leq\!\min\!\left\{ R_{i}^{\max
},U_{i}\left( t\right)\!-\!\sum\limits_{j=1}^{J_{i}}r_{ij}\left( t\right)
\!\right\} \notag
\end{align}
\vspace{-1.5cm}
\begin{align}
\textbf{P6:}~~&\min_{\!D_i^d(t)}~~\phi_t(\overrightarrow{D^d},\overrightarrow{\lambda})
\!=\!\sum_{i=1}^I\!D_i^d(t)\!\left(\!H_i(t)\eta^-+Vc(t)\!\right)\!+\!\sum_{l=1}^L\!\lambda_l
\!\left(\sum_{i:X_{li}=1}D_i^d(t)\!-\!P_l\!+\!\textbf{E}\{N_l(t)\}\!+\!\frac{\sigma_{N_l}}{\sqrt{2\epsilon}}\!\right) \\ \label{eq_21}
&\mbox{s.t.}~~D_{i}^{d}\left( t\right)\!\in\!\left[ 0,\left(
\sum\limits_{j=1}^{J_{i}}r_{ij}\left( t\right)\!-\!U_{i}\left( t\right) \right)
^{+}\right] \notag
\end{align}

Notice that \textbf{P4(a)}, \textbf{P4(b)}, and \textbf{P5} are independent of Lagrangian multipliers. Thus only the dual problem of \textbf{P6} shall be considered as follows.
\vspace{-0.5cm}
\begin{align}
\widehat \textbf{P6}\textbf{:}~~&\max_{\overrightarrow{\lambda}}~\min_{D_i^d(t)}~~\phi_t(\overrightarrow{D^d},\overrightarrow{\lambda}) \\ \label{eq_22}
&\mbox{s.t.}~~\lambda_l\geq 0~~~\forall l\in \textbf{L}, \notag
\end{align}which is equivalent to:
\vspace{-0.5cm}\begin{equation}\label{eq_23}
\begin{small}
\max_{\overrightarrow{\lambda}}\{\sum_{l=1}^L\lambda_l(\textbf{E}\{N_l(t)\}+\frac{\sigma_{N_l}}{\sqrt{2\epsilon}}-P_l)+ \min_{\begin{subarray}{1} D_i^b(t)\end{subarray}}\sum_{i=1}^I\Lambda_i(D_i^d(t),\overrightarrow{\lambda})\},
\end{small}\end{equation}where $\Lambda_i(D_i^d(t),\overrightarrow{\lambda})=D_i^d(t)(H_i(t)\eta^-+Vc(t)+\sum_{l:X_{li}=1}\lambda_l)$.
\vspace{-0.5cm}
\subsection{Directing Policy and Distributed Energy Allocation Laws}
In this subsection, the solutions of \textbf{P3}$\sim$\textbf{P6} will be given. The whole algorithm consists of directing policy and energy allocation laws. Directing policy, the solution of \textbf{P3}, decides which outlet that the PEV registered at the $k$-th entry point should be charged in. Distributed energy allocation laws, implemented by each charging station, are solutions of \textbf{P4}$\sim$\textbf{P6}.
\subsubsection{Directing Policy}
After every PEV uploads their energy demands via entry points, and every charging station uploads their battery state, the central controller implements the following Algorithm \ref{alg_1}.
\renewcommand{\algorithmicrequire}{\textbf{Input:}} 
\renewcommand{\algorithmicensure}{\textbf{Output:}} 
\vspace{-0.5cm}
\begin{algorithm}
\caption{Directing Policy}\label{alg_1}
\begin{algorithmic}
\REQUIRE $E_k(t),H_i(t)\eta^-$
\ENSURE The decision variable $W_{k,ij}(t)$
\STATE Define two sets: $\textbf{Y}\!\!=\!\!\{k|E_k(t)\!>\!0\}$ and $\textbf{Z}\!\!=\!\!\{(i,j)|Q_{ij}\!=\!0\}$ to include the nonempty entry points and empty outlets in time slot $t$ respectively.
\STATE Let all $W_{k,ij}(t)=0$.
\WHILE {$\textbf{Y}\neq \O$ \&\& $\textbf{Z}\neq \O$}
\STATE Choose the entry point having the largest demand $E_k(t)$ from \textbf{Y}.
\STATE Search for an empty outlet $(i,j)$ among \textbf{Z} to get the maximum of $E_k(t)+H_i(t)\eta^-$.
\STATE Let $W_{k,ij}(t)=1$, remove $k$ from \textbf{Y} and $(i,j)$ from \textbf{Z}.
\ENDWHILE
\end{algorithmic}
\end{algorithm}
\vspace{-0.3cm}
\subsubsection{Charging Rate of PEV}
For every charging outlet, the charging rate is determined  by solving problems \textbf{P4(a)} and \textbf{P4(b)}, case by case.

$a)$ For outlet that is non-empty ($Q_{ij}(t)>0$) at the beginning of $t$: The charging rate depends on the demand queue state and battery state.

$b)$ For outlet that is empty ($Q_{ij}(t)=0$) at the beginning of $t$: The charging rate depends on the coming demand in this time slot and also battery state.

The charging rate control law shown in Fig. 2 means to serve the demand as fast as possible when the battery has sufficient power.
\begin{figure}
\begin{minipage}[t]{0.5\textwidth}
\centering
\includegraphics[width=80mm]{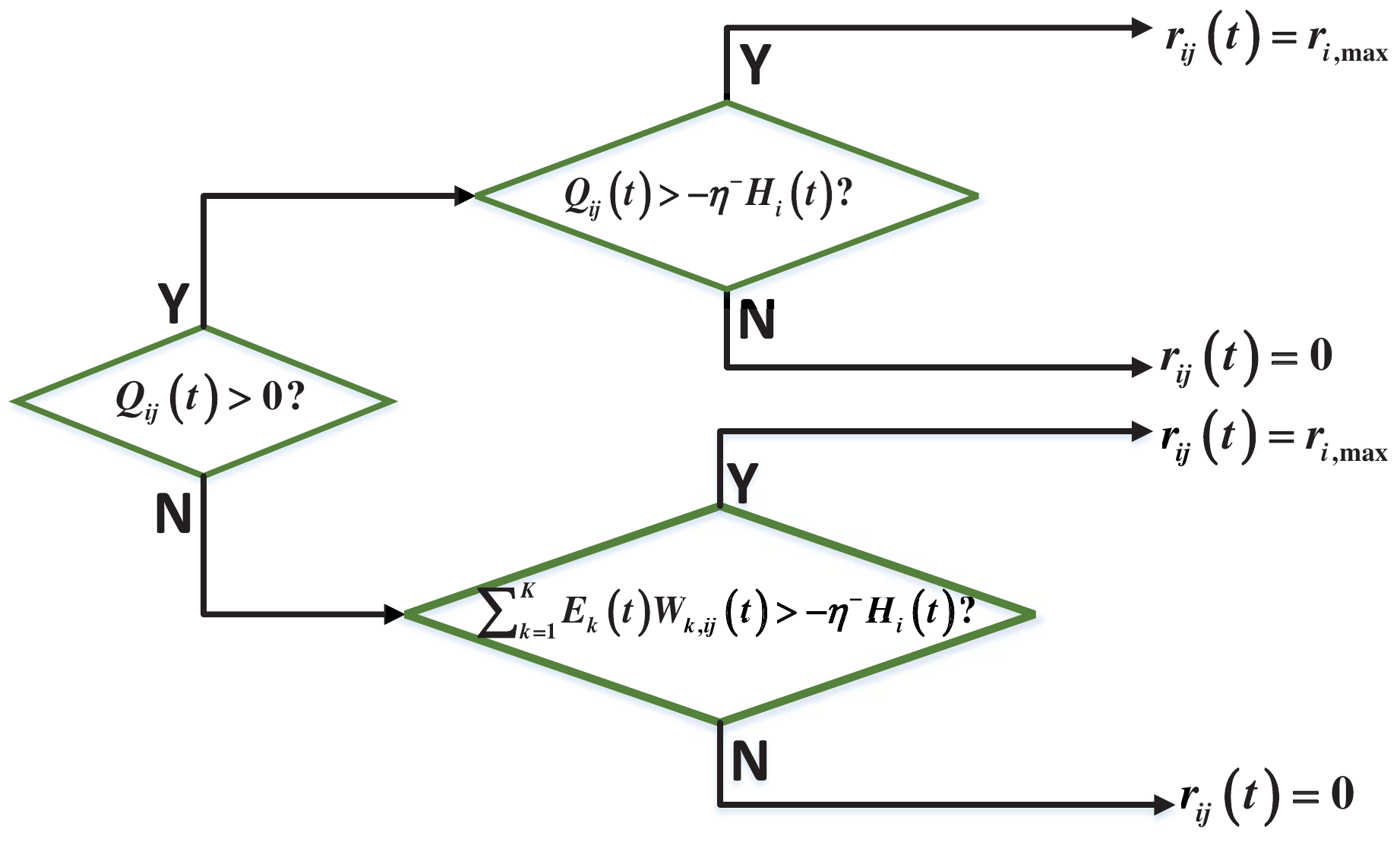}
\caption{Charging rate control law.}
\label{fig_2}
\end{minipage}%
\begin{minipage}[t]{0.5\textwidth}
\centering
\includegraphics[width=80mm]{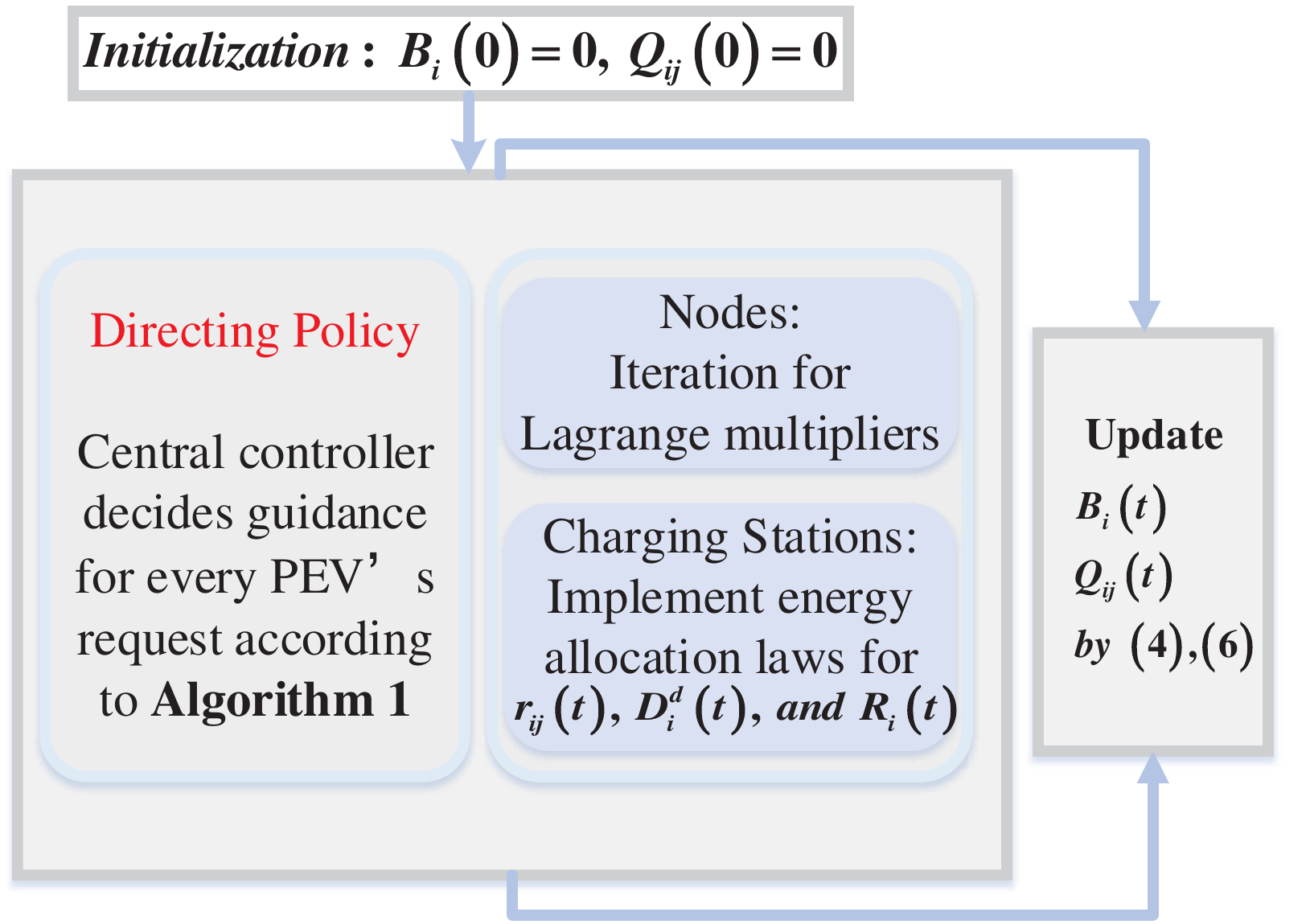}
\caption{A sketch for the whole algorithm.}
\label{fig_3}
\end{minipage}
\end{figure}
\subsubsection{Renewable energy input}
Every charging station determines the charging rate of their harvested renewable energy at time slot $t$, depending on the battery state, also in a distributed way,
\vspace{-0.3cm}
\[
R_i(t)
= \left\{
\begin{aligned}
\min\{R_{i}^{max},(U_{i}(t)-\sum_{j=1}^{J_{i}}r_{ij}(t))^+\}~~~~\\
0~~~~\\
\end{aligned}
\right.
\begin{aligned}
&\text{if}~ H_{i}(t) \leq 0 \\
&\text{otherwise}~~ ~~
\end{aligned}\]

It is intuitive that $H_i(t)<0$ means that the battery is thirsty so that all available renewable energy should be charged into it; however, when there is enough electricity in the battery, no renewable energy will be charged into, for the sake of battery constraints in (\ref{eq_5}).
\subsubsection{Consumption from the power grid}
By introducing the Lagrangian dual problem $\widehat \textbf{P6}$, each charging station can locally solve a subproblem given by
\vspace{-0.5cm}
\begin{equation}\label{eq_24}
\begin{array}{cc}
\min & \Lambda _{i}\left( D_{i}^{d}\left( t\right) ,\vec{\lambda}\right) \\
\text{s.t.} & 0\!\leq\!D_{i}^{d}\left( t\right)\!\leq\!\left(
\sum\limits_{j=1}^{J_{i}}r_{ij}\left( t\right)\!-\!U_{i}\left( t\right) \right)
^{+}.%
\end{array}
\end{equation}

The dual problem is solved by using gradient projection method, and the Lagrangian multipliers are updated as following:
\vspace{-0.5cm}
\begin{equation}\label{eq_25}
\begin{split}
&\lambda_l(n\!+\!1,t)\!=\!\left[\lambda_l(n,t)\!-\!\kappa\left(P_l\!-\!\!\!\!\sum_{i:X_{li}=1}\!\!D_i^d(n,t)-\textbf{E}\{N_l(t)\}-\frac{\sigma_{N_l}}
{\sqrt{2\epsilon}}\right)\right]^+,
\end{split}
\end{equation}
where $\kappa$ is a sufficiently small positive constant, known as the step size and $n$ is the iteration index in time slot $t$.

Since $\Lambda_i$ is a linear function of $D_i^d(t)$, the optimum of (\ref{eq_24}) depends on the monotonicity of $\Lambda_i$. Algorithm 2 states how to choose $D_i^d(t)$, where $M$ is the greatest number of iterations and $\xi$ is the convergence criterion. The following result characterizes the effects of constant step size on the convergence property.
\vspace{-0.8cm}
\begin{Lemma}
By the above iteration (\ref{eq_25}) with constant stepsize, the dual function (22) is guaranteed to
converge to a suboptimal value within a finite number of steps,
\textit{i.e., }$\varphi _{t}^{\text{*}}\left( \vec{\lambda}\right)\!-\!\varphi _{n,t}^{%
\text{best}}\left( \vec{\lambda}\left( m,t\right) \right)\!<\!\kappa
L\max_{l}\left\{ P_{l}^{2}\!+\!\left( D_{l}^{\max }\!+\!\mathbf{E}\left(
N_{l}\right)\!+\!\frac{\sigma _{N_{l}}}{\sqrt{2\epsilon }}\right) ^{2}\right\} $
within at most $\frac{\left\Vert \vec{\lambda}\left( 0,t\right)\!-\!\vec{\lambda%
}^{\ast }\left( t\right) \right\Vert _{2}^{2}}{\kappa ^{2}L\max_{l}\left\{
P_{l}^{2}\!+\!\left( D_{{}}^{\max }\!+\!\mathbf{E}\left( N_{l}\right)\!+\!\frac{\sigma
_{N_{l}}}{\sqrt{2\epsilon }}\right) ^{2}\right\} }$ steps, where $\varphi
_{t}^{\ast }\left( \vec{\lambda}\right) $ is the optimal value of (22)
and $\varphi _{n,t}^{\text{best}}\left( \vec{\lambda}\left( m,t\right)
\right) \hat{=}\max_{m\in \left\{ 1,\cdots ,n\right\} }\varphi _{m,t}\left(
\vec{\lambda}\left( m,t\right) \right) $ is the best value of (22)
till now, $\vec{\lambda}\left( 0,t\right) $ is the initial value of
Lagrangian multiplier vector and $\vec{\lambda}^{\ast }\left( t\right) $ is the
optimal value vector. The above lemma shows that the stepsize can be used to tradeoff convergence accuracy and speed.
Due to space limitation, the detailed proof can be found in \cite{Boyd}.
\end{Lemma}
\begin{algorithm}
\caption{Consumption From the Power Grid}\label{alg_2}
\begin{algorithmic}[1]
\REQUIRE $H_i(t),c(t),\eta^-$
\ENSURE $D_i^d(t)$
\STATE Let all $\lambda_l=\lambda_{max}$.
\WHILE {$\left|\phi_t|_{\overrightarrow{\lambda}\!(\!n\!)}-\phi_t|_{\overrightarrow{\lambda}\!(\!n-1\!)}\right|\geq \xi ~$\&\&$~ n < M$}
\STATE Compute new values of $\overrightarrow{\lambda}$ through (24).
\IF {$\frac{\partial\Lambda_i}{\partial D_i^d(n,t)}=\frac{\partial\Lambda_i}{\partial D_i^b(n,t)} >0$}
\STATE $D_i^d(n,t)=0$
\ELSE
\STATE $D_i^d(n,t)\!\!=\!\!\left( \sum\limits_{j=1}^{J_{i}}r_{ij}\left( t\right) -U_{i}\left( t\right)\right) ^{+}$
\ENDIF
\STATE $n=n+1$
\ENDWHILE
\STATE $D_i^d(t)=D_i^d(n,t)$
\end{algorithmic}
\end{algorithm}

The whole algorithm  for the problem $\widetilde\textbf{P1}$ is given in Fig. \ref{fig_3}. At the beginning of time slot $t$, every entry point uploads the energy demand arriving at them, and every charging station uploads the states of demand queues and batteries, for the central controller to guide different PEVs to charge at different charging outlets. At the same time each node in the distribution network computes the Lagrangian multipliers as the overload indication through (\ref{eq_25}) and broadcasts them to downstream charging stations. When receiving the Lagrange multipliers, each charging station calculates the charging rate, the utilization of renewable energy, and the energy consumption from the power grid.
\vspace{-0.3cm}
\section{Performance Analysis}
Stability of battery dynamics and demand queues should be guaranteed in the above algorithms. Also the performance gap between the proposed algorithm and the optimal one in solving \textbf{P1} is given.
\vspace{-0.5cm}
\subsection{Stability of Queues}
\begin{Theorem} \label{thm_1}
All demand queues at every charging outlet are stable.
\end{Theorem}
\begin{proof}
Since only one EV can be served by each outlet, $Q_{ij}(t) \leq E_{\max}$ for all time slot $t$.
\end{proof}
\begin{Remark}\label{rem_1}
This inequality guarantees the worst delay for every PEV. By employing virtual queue as \cite{Jin}, the length of worst delay can be calculated and limited easily.
\end{Remark}
\begin{Theorem}\label{thm_2}
For any parameter $V$ bounded between  $0\!<\!V\!\leq\!V_{\max}$ at any time slot $t$ , where
\vspace{-0.5cm}\begin{equation} \label{eq_26}\begin{split}&
V_{\max}\!=\!\frac{\min\left\{ B_{i,\max }\right\}\!-\!\eta ^{+}U_{\max }\!-\!\eta
^{-}Jr_{\max }\!-\!L\lambda _{\max }}{C_{\max }},
\end{split}\end{equation}$T_{\max}\!\!=\!\!VC_{\max}\!\!+\!\!L\lambda_{\max}$, $J_{\max}\!\!=\!\!\max_i\{J_i\}$,
the algorithm satisfies: $0\!\!\leq\!\!B_{i}(t)\!\leq\!B_{imax}$. Before the proof of Theorem 2, the following lemma is given first.
\end{Theorem}
\begin{Lemma}
The energy management algorithm developed by solving \textbf{P5} and \textbf{%
P6 }can ensure the following conditions, 1) If $B_{i}\left( t\right)\!\!\geq
\!\!B_{max}\!\!-\!\!\eta ^{+}U_{\max },$ then $R_{i}\left( t\right)\!\!=\!\!0;$ 2) If $
B_{i}\left( t\right)\!\!\leq\!\!\eta ^{-}J_{i}r_{i\max },$ then $D_{i}^{d}\left(
t\right)\!\!=\!\!\sum_{j=1}^{J_{i}}\!\!\left( r_{ij}\!\!-\!\!U_{i}\left( t\right)\!\right) ^{+}.$
\end{Lemma}

\begin{proof}
If $B_{i}\left( t\right)\!\geq\!B_{max}\!-\!\eta ^{+}U_{\max }$, we have $B_{i}\!\left(t\right)\!-\!\left( J_{i}r_{i\max}\eta^{-}\!+\!\left(B_{max}\!\!-\!\!\eta
^{+}U_{\max}\!\!-\!\!\eta ^{-}Jr_{\max }\!\!-\!\!L\lambda _{\max}\right)\!\!+\!\!L\lambda_{\max}\!\right)\!\geq \!0.$
Due to the definition of $V_{\max }$, it is obtained that $%
B_{i}\left( t\right) -\left( J_{i}r_{i\max }\eta ^{-}+V_{\max }C_{\max
}+L\lambda _{\max }\right) \geq 0$, i.e., $H_{i}\left( t\right)>0$.
Then we have $R_{i}\left( t\right) =0$ according to the solution to \textbf{P5.}

If $B_{i}\left( t\right)\!\leq\!\eta ^{-}J_{i}r_{i\max },$ we have $B_{i}\left( t\right)\!-\!\left(\eta ^{-}J_{i}r_{i,max }\!+\!VC_{max }\!+\!L\lambda_{max }\!\right)\!\leq\!-\!\left( VC_{max }\!+\!L\lambda_{max }\!\right)$. Then we have $H_{i}\left(t\right)\!<\!-\!\left( VC_{max }\!+\!L\lambda _{max }\!\right)$, which results in $H_{i}\left(t\right)\!<\!-\!\frac{Vc\left(t\right)\!+\!\sum_{l:X_{li}=1}\lambda _{l}\left( t\right) }{\eta^{-}}$. Thus, $D_{i}^{d}\left(t\right)\!=\!\sum_{j=1}^{J_{i}}\!\left(r_{ij}\!-\!U_{i}\left(t\right)\!\right)^{+}$ is obtained by solving \textbf{P6.}
\end{proof}

Now we give the proof of Theorem 2.
\begin{proof}
The bound of $B_i(t)$ will be proved by induction. At the beginning, the energy level of each battery satisfies:
\vspace{-0.5cm}
\begin{equation*}
0 \leq B_i(0)\leq B_{i,{max}}.
\end{equation*}
\vspace{-1.5cm}
\begin{flushleft}
$a)$ When $T_{max}+\eta^+J_ir_{i,max}\leq B_i(t)\leq B_{i,{max}}$,
$R_{i}\left( t\right) =0$ holds according to the solution to \textbf{P5.} Then
\vspace{-0.3cm}
\begin{equation*}
\begin{small}
B_{i}\!\left( t+1\!\right)\!=\!B_{i}\left( t\right)\!-\!\eta ^{-}\!\left(\!\left(
\sum\limits_{j=1}^{J_{i}}r_{ij}\left( t\right)\!-\!U_{i}\left( t\right)\!\right)
^{+}\!-\!D_{i}^{d}\left( t\right)\!\right)\!\leq\!B_{i}\left( t\right)\!\leq
\!B_{i,\max}.
\end{small}
\end{equation*}
\end{flushleft}
\vspace{-0.7cm}
\begin{flushleft}
$b)$ When $H_i(t)\leq T_{max}\!+\!\eta^+J_ir_{i,max}$, we have $R_{i}\left( t\right)\!=\!\min\!\left\{ R_{i}^{\max },\left( U_{i}\left(
t\right)\!-\!\sum\limits_{j=1}^{J_{i}}r_{ij}\left( t\right)\!\right)^{+}\!\right\}.$ Therefore, it is obtained that
\vspace{-0.5cm}
\begin{eqnarray}
B_{i}\!\left( t+1\!\right)\!\!&\leq&\!\!B_{i}\left( t\right) +\eta ^{+}R_{i}\left(
t\right) <\eta ^{-}Jr_{\max }+VC_{\max }+L\lambda _{\max }+\eta ^{+}\left(
U_{i}\left( t\right) -\sum\limits_{j=1}^{J_{i}}r_{ij}\left( t\right) \right)
^{+}  \nonumber \\
&\leq&\!\!\eta ^{-}Jr_{\max }+VC_{\max }+L\lambda _{\max }+\eta ^{+}U_{\max }
\nonumber \\
&\leq &\!\!\eta ^{-}Jr_{\max }+\min \left\{ B_{i,\max }\right\} -\eta
^{+}U_{\max }-\eta ^{-}Jr_{\max }-L\lambda _{\max }+L\lambda _{\max }+\eta
^{+}U_{\max }  \\ \label{eq_27}
&\leq &\!\!B_{i,\max},  \nonumber
\end{eqnarray}
where (27) is due to the definition of $V_{\max}.$
\end{flushleft}
\vspace{-0.5cm}
\begin{flushleft}
$c)$ When $0\leq B_{i}\left( t\right) \leq \eta ^{-}Jr_{\max },$ due to the
second statement in Lemma 1, we have
\vspace{-0.5cm}
\begin{equation*}
\begin{small}
D_{i}^{d}\left( t\right)\!=\!\left(\sum\limits_{j=1}^{J_{i}}r_{ij}\left( t\right)\!-\!U_{i}\left( t\right)\!\right)^{+},
\end{small}
\end{equation*}which means that the $i$-th station uses grid energy to supply the charging request totally. Therefore, by the dynamics of battery, we have
\vspace{-0.5cm}
\begin{equation*}
\begin{small}
B_{i}\!\left( t+1\!\right)\!\geq\!B_{i}\left( t\right)\!\geq\!0.
\end{small}
\end{equation*}
\vspace{-0.5cm}
\end{flushleft}
\vspace{-1.4cm}
\begin{flushleft}
$d)$ When $B_{i}\left( t\right) \geq \eta ^{-}Jr_{\max },$ we have $R_{i}\left(
t\right) $ according to Lemma 1. Then we have
\vspace{-0.3cm}
\begin{equation*}
\begin{small}
B_{i}\!\left( t+1\!\right)
\!=\!B_{i}\left( t\right)\!-\!\eta ^{-}\!\left(\!\left(
\sum\limits_{j=1}^{J_{i}}r_{ij}\left( t\right)\!-\!U_{i}\left( t\right)\!\right)
^{+}\!-\!D_{i}^{d}\left( t\right)\!\right)\!\geq\!B_{i}\left( t\right)\!-\!\eta
^{-}Jr_{\max }\!\geq\!0.
\end{small}
\end{equation*}
\end{flushleft}

Thus for all time slot $t$, the conclusion in Theorem 2 stands.
\end{proof}
Theorem 2 means that the constraint (\ref{eq_5}) is satisfied at every time slot although the proposed algorithm in Fig. \ref{fig_3} solves $\widetilde\textbf{P1}$. The maximum of Lagrangian multipliers ($\lambda_{max}$) exists because the gradient method in Algorithm 2 converges.
\vspace{-0.5cm}
\subsection{Asymptotic Performance}
Before we state the performance of the proposed algorithm the following Lemma is needed.
\begin{Lemma}
(stationary and randomized policy): If all $E_k(t)$, $U_i(t)$, $c(t)$ are \emph{independent and identically distributed (i.i.d.)} over time slots, then taking all control decisions $\widehat{r_{ij}(t)}$, $\widehat{D_{i}^d(t)}$, and $\widehat{R_{i}(t)}$ every time slot $t$  only as a function of current system state $\omega(t)=\{E_k(t),U_i(t),c(t)|\forall k\!\!\in\!\!\textbf{K}$, $\forall i\!\!\in\!\!\textbf{I}\}$, the constraints mentioned in problem $\widetilde \textbf{P1}$ are satisfied and the following properties stand:
\vspace{-0.3cm}
\begin{equation} \label{eq_28}
\begin{small}
\mathbf{E}\!\left\{ \eta ^{-}\!\left(\!\sum\limits_{j=1}^{J_{i}}\widehat{%
r_{ij}\left( t\right) }\!-\!\widehat{U_{i}\left( t\right) }\!\right) ^{+}\!\right\}\!-\!%
\mathbf{E}\!\left\{ \eta ^{-}\widehat{D_{i}^{d}\left( t\right) }\!\right\}\!=\!%
\mathbf{E}\!\left\{ \eta ^{+}\widehat{R_{i}\left( t\right) }\!\right\},
\end{small}
\end{equation}
\vspace{-1.2cm}
\begin{equation} \label{eq_29}
\textbf{E}\{c(t)\sum_{i=1}^N\widehat{D_{i}^d(t)}\} = P_{re}^\ast,
\end{equation}where $P_{re}^\ast$ is the optimal value of problem $\widetilde \textbf{P1}$, and the expectations are with respect to the stationary distribution of $\omega(t)$ and randomized control decisions.
\end{Lemma}
\begin{proof}
Similar policy is mentioned in the work of Fang \emph{et al.} in \cite{Guo}. and also proved by Neely \emph{et al.} in \cite{Georgiadis} and \cite{Neely2}. So it is omitted here.
\end{proof}
\begin{Theorem}\label{thm_3}
If $E_k(t)$, $U_i(t)$, $c(t)$  are \emph{i.i.d.} over slots, an upper bound for the time average expected electricity cost under our algorithm can be obtained and the upper bound is as follows:
\vspace{-0.5cm}
\[\lim_{\tau\rightarrow \infty}\frac{1}{\tau}{\sum}_{t=0}^{\tau-1}\textbf{E}\{c(t)\sum_{i=1}^ID_{i}^d(t)\} \leq P^{\ast} + \frac{\delta_{\max}}{V}\].
\end{Theorem}
\vspace{-1.2cm}
\begin{proof}
From (\ref{eq_15}), the following inequality is tenable:
\vspace{-0.3cm}
\begin{equation} \label{eq_30}
\begin{split}
&\triangle(\overrightarrow{\theta(t)})\!\!+\!\!V\textbf{E}\{c(t)\!\!\sum_{i=1}^{I}D_i^d(t)]|\overrightarrow{\theta(t)}\!\}\!\leq\! \delta_{max}\!-\!\sum\limits_{i=1}^{I}\!\mathbf{E}\!\left\{\sum\limits_{j=1}^{J_{i}}r_{ij}\left( t\right) Q_{ij}\left( t\right) \cdot
1_{\left\{ Q_{ij}\left( t\right) >0\right\} }\!\left\vert \overrightarrow{%
\theta \left( t\right) }\!\right. \!\right\} \\
&\!-\!\sum\limits_{i=1}^{I}\!\mathbf{E}\!\left\{\!\sum\limits_{j=1}^{J_{i}}r_{ij}\left( t\right)
\!\sum\limits_{k=1}^{K}\!E_{k}\left( t\right) W_{k,ij}\left( t\right) \cdot
1_{\left\{ Q_{ij}\left( t\right) =0\right\} }\left\vert \overrightarrow{%
\theta \left( t\right) }\right. \right\}\!+\!\sum\limits_{i=1}^{I}\!\mathbf{E}\!\left\{ H_{i}\left( t\right) R_{i}\left(
t\right) \eta ^{+}\!\left\vert \overrightarrow{\theta \left( t\right) }\!\right.
\!\right\} \\
&\!-\!\sum\limits_{i=1}^{I}\!\mathbf{E}\!\left\{ H_{i}\left( t\right) \eta
^{-}\!\left(\!\left(\!\sum\limits_{j=1}^{J_{i}}r_{ij}\left( t\right)
\!-\!U_{i}\left( t\right)\!\right) ^{+}\!-\!D_{i}^{d}\left( t\right)\!\right)
\left\vert \overrightarrow{\theta \left( t\right) }\right.\!\right\} +V%
\!\mathbf{E}\!\left\{ c\left( t\right)\!\sum\limits_{i=1}^{I}\!D_{i}^{d}\left(
t\right)\!\left\vert\!\overrightarrow{\theta \left( t\right) }\!\right. \right\},
\end{split}\end{equation}since $-T_{max}\!-\!\eta^-J_{i}r_{i,max}\!\leq\!H_i(t)\!\leq\!B_{i,max}\!-\!T_{max}\!-\!\eta^-J_{i}r_{i,max}$ and $B_{i,max}\!>\!T_{max}\!+\!\eta^-J_{i}r_{i,max}.$

Solutions to \textbf{P3}$\sim$\textbf{P6} are willing to minimize the R.H.S of (\ref{eq_30}) at each time slot $t$ by choosing control decisions among all feasible actions including the optimal one given in Lemma 1. Considering that the policy stated in Lemma 1 is independent on the queues vector  $\overrightarrow{\theta(t})$, substitute the optimal decisions $(\widehat {r_{ij}(t)},\widehat {D_{i}^d(t)},\widehat {R_{i}(t)})$ and the conclusion of (\ref{eq_28}) and (\ref{eq_29}) in Lemma 1 into (\ref{eq_30}), and use $P_{re}^{\ast}\leq P^{\ast}$:
\vspace{-0.3cm}\begin{equation*}\begin{split}
&\triangle(\overrightarrow{\theta(t)})+V\textbf{E}\{c(t)\sum_{i=1}^{I}D_i^d(t)|\overrightarrow{\theta(t)}\}\\ &\leq\!\!\delta_{max}\!+\!\sum\limits_{i=1}^{I}\mathbf{E}\left\{ H_{i}\left( t\right) \widehat{R_{i}\left( t\right) }\eta ^{+}\overrightarrow{\theta \left(t\right)}\right\}\!\!-\!\!\sum\limits_{i=1}^{I}\mathbf{E}\left\{
H_{i}\left( t\right) \eta ^{-}\left( \left( \sum\limits_{j=1}^{J_{i}}%
\widehat{r_{ij}\left( t\right) }\!-\!\widehat{U_{i}\left( t\right) }\right)^{+}\!\!-\!\!
\widehat{D_{i}^{d}\left( t\right) }\right) \overrightarrow{\theta
\left( t\right) } \right\} \\
&\!+\!V\mathbf{E}\!\left\{ c\left( t\right)\sum\limits_{i=1}^{I}\widehat{D_{i}^{d}\left( t\right) }\left\vert
\overrightarrow{\theta \left( t\right) }\right. \!\right\}\!=\!\delta_{max}+V\textbf{E}\{c(t)\widehat{D_{i}^d(t})|\overrightarrow{\theta(t})\}
\!=\!\delta_{max}\!+\!VP_{re}^{\ast}
\!\leq\!\delta_{max}\!+\!VP^{\ast}.\end{split}\end{equation*}

Sum over $t \in \{0,1,...,\tau-1\}$ , divide $V\tau$ on both sides, and let $\tau\rightarrow \infty$:
\vspace{-0.5cm}
\[\lim_{\tau\rightarrow \infty}\frac{1}{\tau}{\sum}_{t=0}^{\tau-1}\textbf{E}\{c(t)\sum_{i=1}^ID_{i}^d(t)\} \leq P^{\ast} + \frac{\delta_{\max}}{V}\]holds due to the fact that $\textbf{E}\{\mathcal{L}(\overrightarrow{\theta(\tau)})\}$ is nonnegative and $\textbf{E}\{\mathcal{L}(\overrightarrow{\theta(0)})\}$ is finite.
\end{proof}
\begin{Remark}\label{rem_2}
It is noted that the performance gap between the proposed algorithm and the optimal one for \textbf{P1} can be decreased by an increasing $V$, which means the capacity of each battery should be increased according to (\ref{eq_26}).
\end{Remark}
\vspace{-0.5cm}
\section{Simulation}
\vspace{-0.5cm}
\subsection{Simulation Setup}
We evaluate our designed algorithms using Matlab on a 19-bus test feeder \cite{Yunus} shown in Fig. \ref{fig_topo}. The system contains 50 entry points and 18 charging stations, each of which has three charging outlets. All the parameters are set as follows except for other specifications. The capacity of each battery in different stations is 500kWh. Let $r_{1,max}\!=...=\!r_{I,max}\!=\!20$kW, $U_{max}\!=\!225$kW, $D^{d,max}\!=\!20$kW, and $E_{max}\!=\!30$kWh. Assume the uncontrollable loads at each node follow gaussian distributions where the expectation is $\textbf{E}\{N_l(t)\}\!=\!200$kW, and the standard deviation is $\sigma_{N_l}\!=\!100$kW. PEVs charging requests arrive to each entry points according to a geometric distribution with probability 0.9. The charging and discharging efficiency parameters are first assumed to be $\eta^-=\eta^+=1$.
\begin{figure}
\centering
\includegraphics[width=100mm]{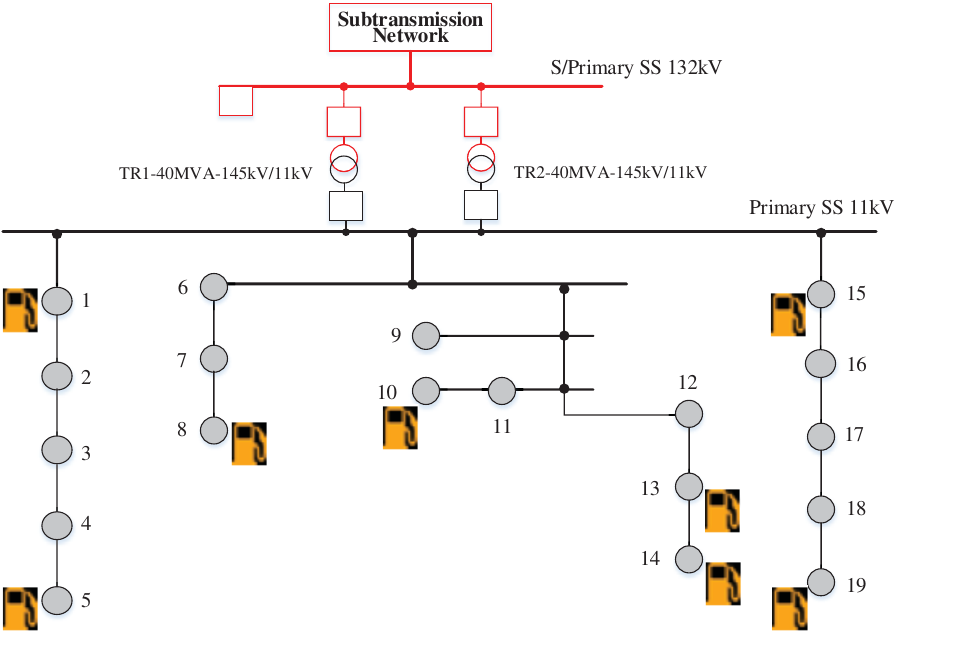}
\caption{A 19-bus bus test feeder.\label{fig_topo}}
\end{figure}
Real wind speed data are taken from Alternative Energy Institute (AEI) \cite{Wind}, and converted to wind energy generation based on the power curve of the Vestas V27 225 kW wind turbine. Fig. \ref{fig_4} shows one wind turbine's generation data. Real-time electricity price data are from Power Smart Pricing administered for Ameren Illinois \cite{Pricing}, which is shown in Fig. \ref{fig_5}. Both of them are updating every one hour. By linear interpolation for electricity price and spline interpolation for wind energy, the interval of these data is 10min, equal to the length of one time slot. The total length of data is 60h.
\begin{figure}
\begin{minipage}[t]{0.5\textwidth}
\centering
\includegraphics[width=70mm]{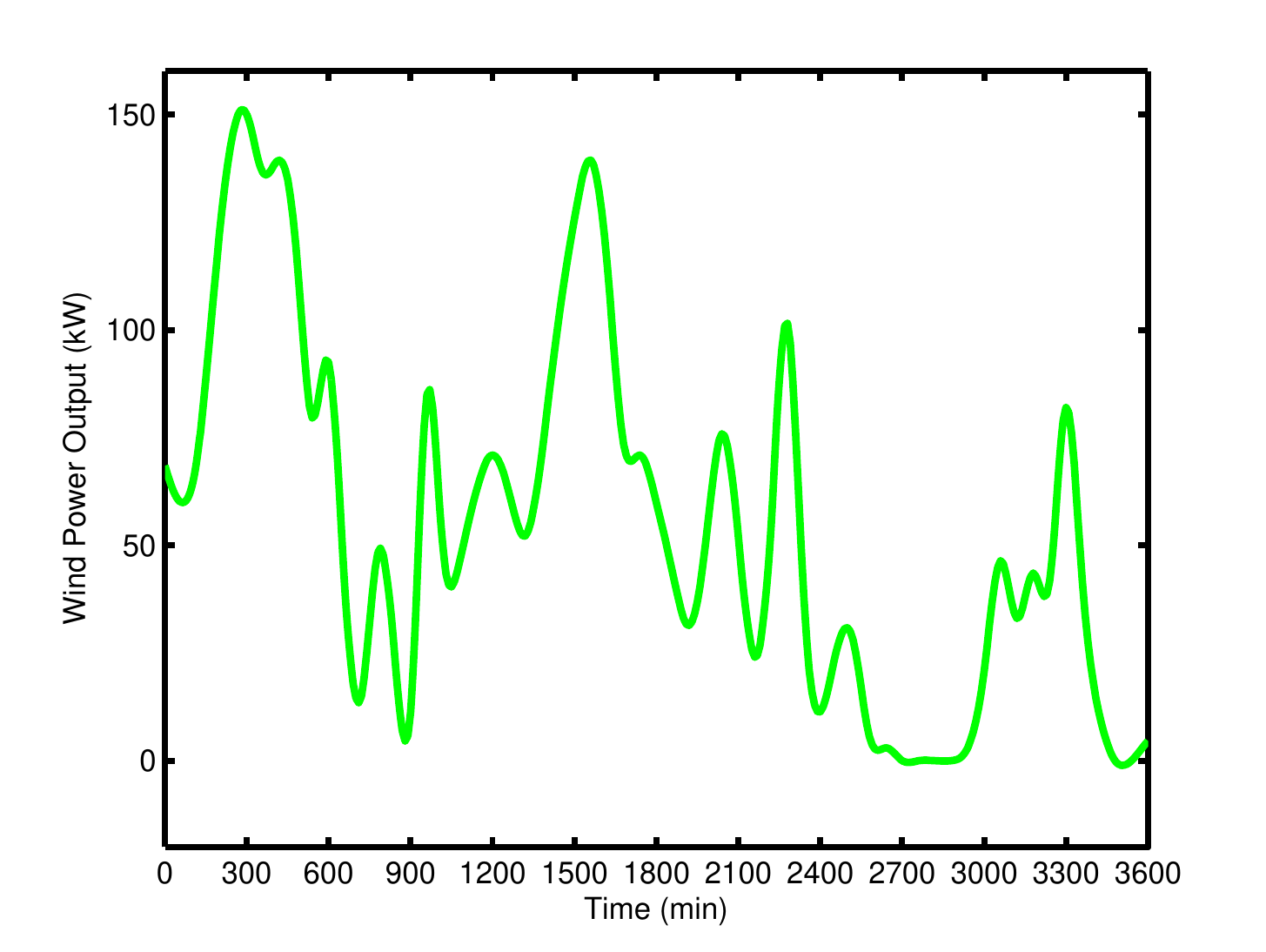}
\caption{Wind power at charging station 7.}
\label{fig_4}
\end{minipage}%
\begin{minipage}[t]{0.5\textwidth}
\centering
\includegraphics[width=70mm]{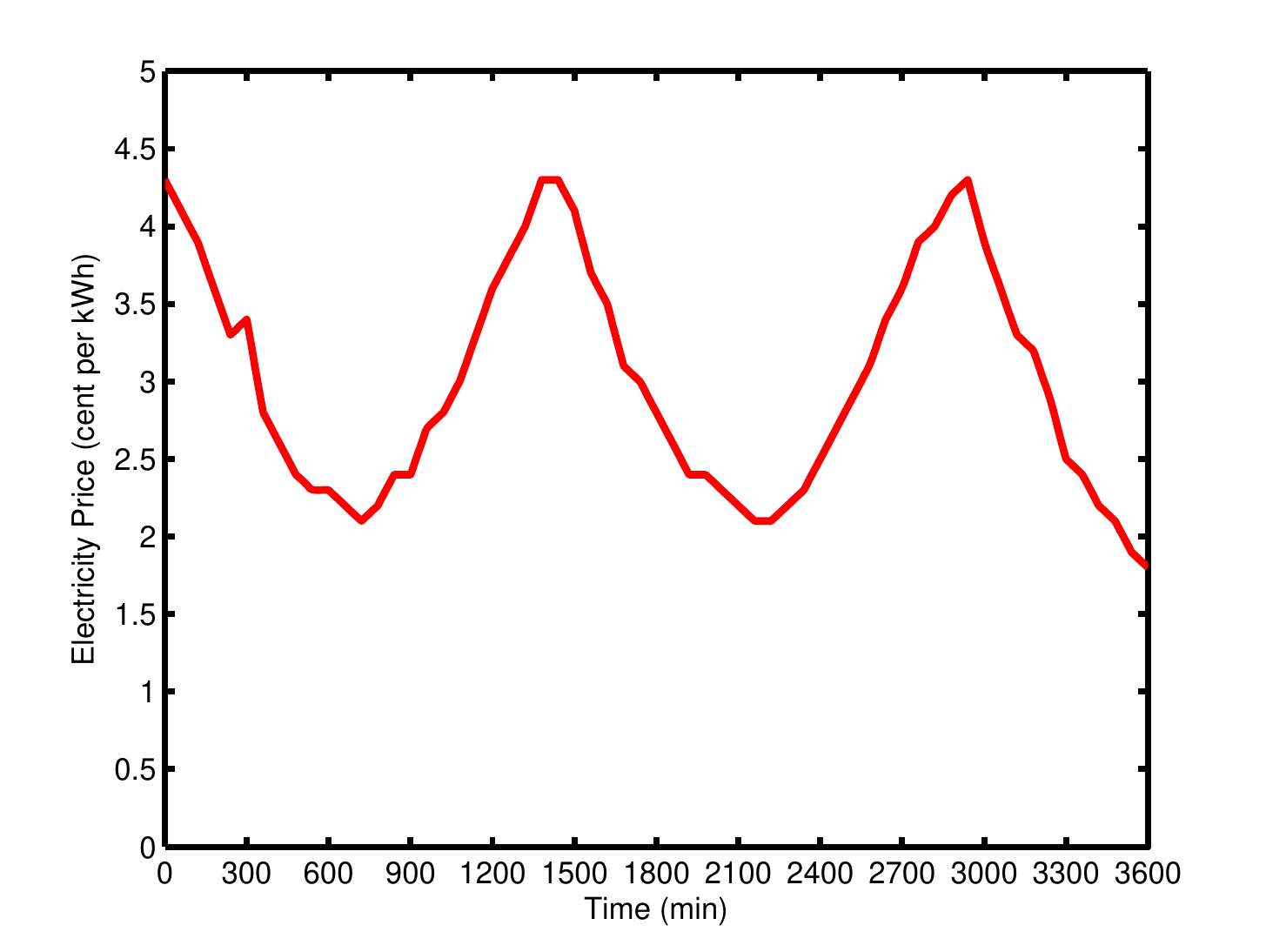}
\caption{Real-time market price.}
\label{fig_5}
\end{minipage}
\end{figure}
\vspace{-0.5cm}
\subsection{Numerical Results}
Fig. \ref{fig_6} illustrates the convergent values for the 8 batteries' energy level under the proposed algorithm. The values are approximately 450kWh, falling in the feasible region of batteries. The existence of convergence for batteries' energy level indicates that the charging and discharging power amounts are roughly equal.
\begin{figure}
\begin{minipage}[t]{0.5\textwidth}
\centering
\includegraphics[width=70mm]{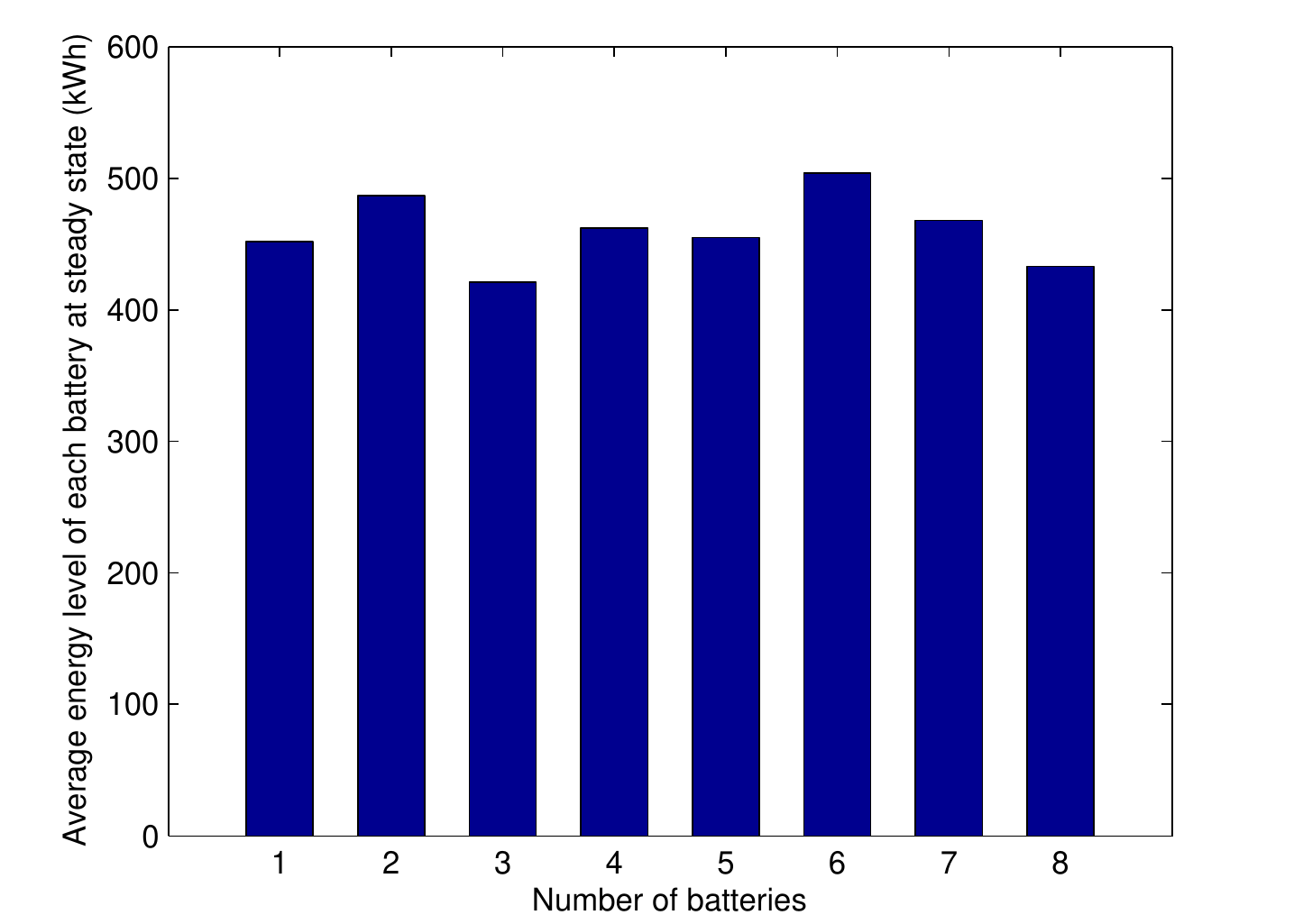}
\caption{Average energy level of each battery at steady state. }\label{fig_6}
\end{minipage}%
\begin{minipage}[t]{0.5\textwidth}
\centering
\includegraphics[width=70mm]{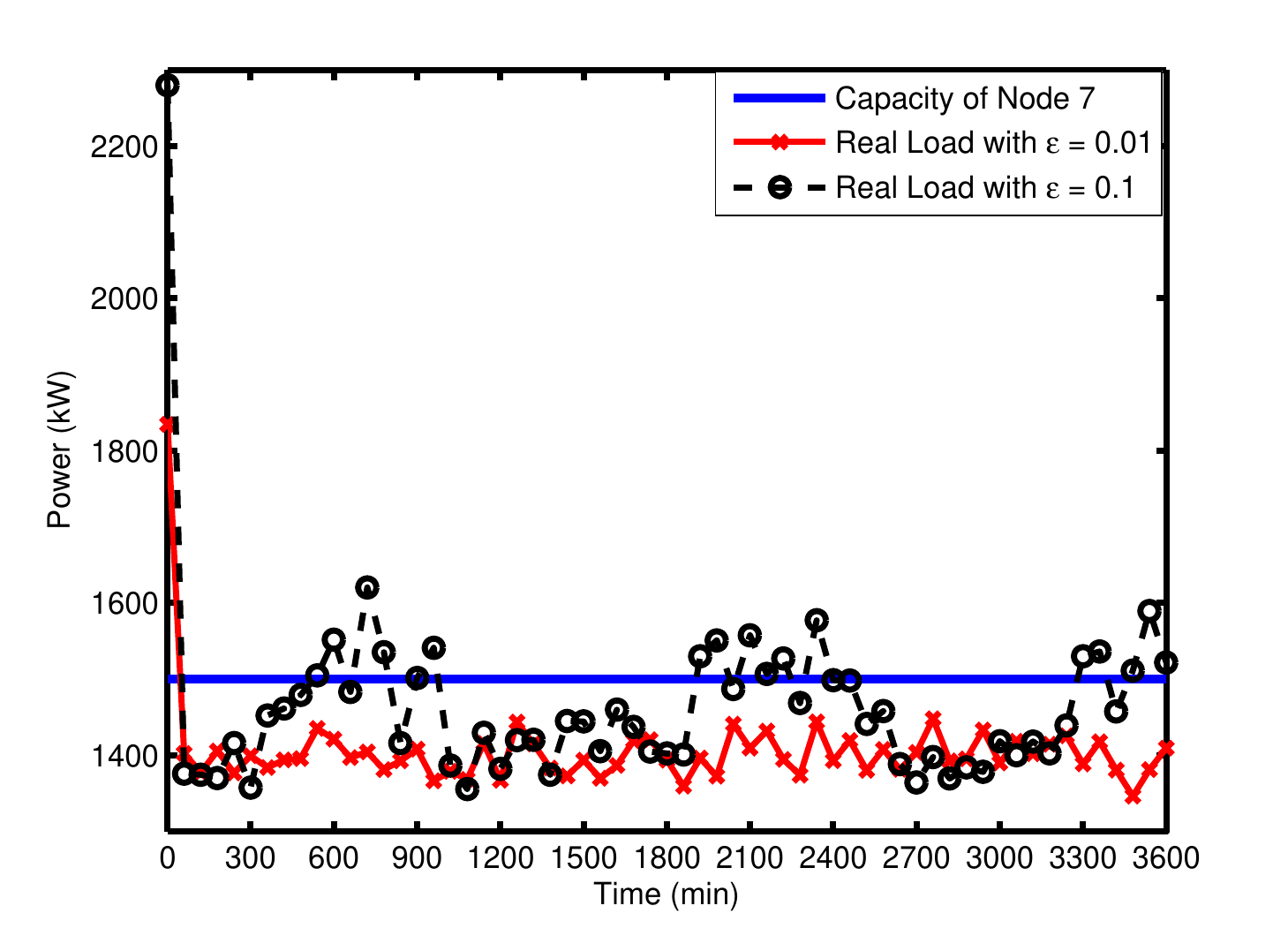}
\caption{Overload: comparison between two different possibilities.}
\label{fig_7}
\end{minipage}
\end{figure}

The probability of overload at node 7 in this distribution network is examined in Fig. \ref{fig_7}. Different values for permitted overload probability $\epsilon$ are taken into account. This figure shows good performance for ``chance constraint'' of the proposed algorithm. When the permitted overload probability is tiny, there is almost no overload possibility for the total loads. Obviously, during 600min-800min, 2100min-2400min, and 3300min-3600min, when the electricity price is low, there are some peaks at the corresponding node.

From Fig. \ref{fig_8}, it is obvious that the system cost rises with respect to the increasing probability of PEV's arrival request at every entry point. Meanwhile, Fig. \ref{fig_8} also illustrates that the larger $V$ is, the less non-renewable energy cost is incurred, which confirms the asymptotically optimal performance of the proposed algorithm shown in Theorem 3. It is also observed that a lower charging efficiency incurs a slightly higher system cost, since the degraded renewable energy utilization results in more on-grid energy consumption.
\begin{figure}
\centering
\includegraphics[width=80mm]{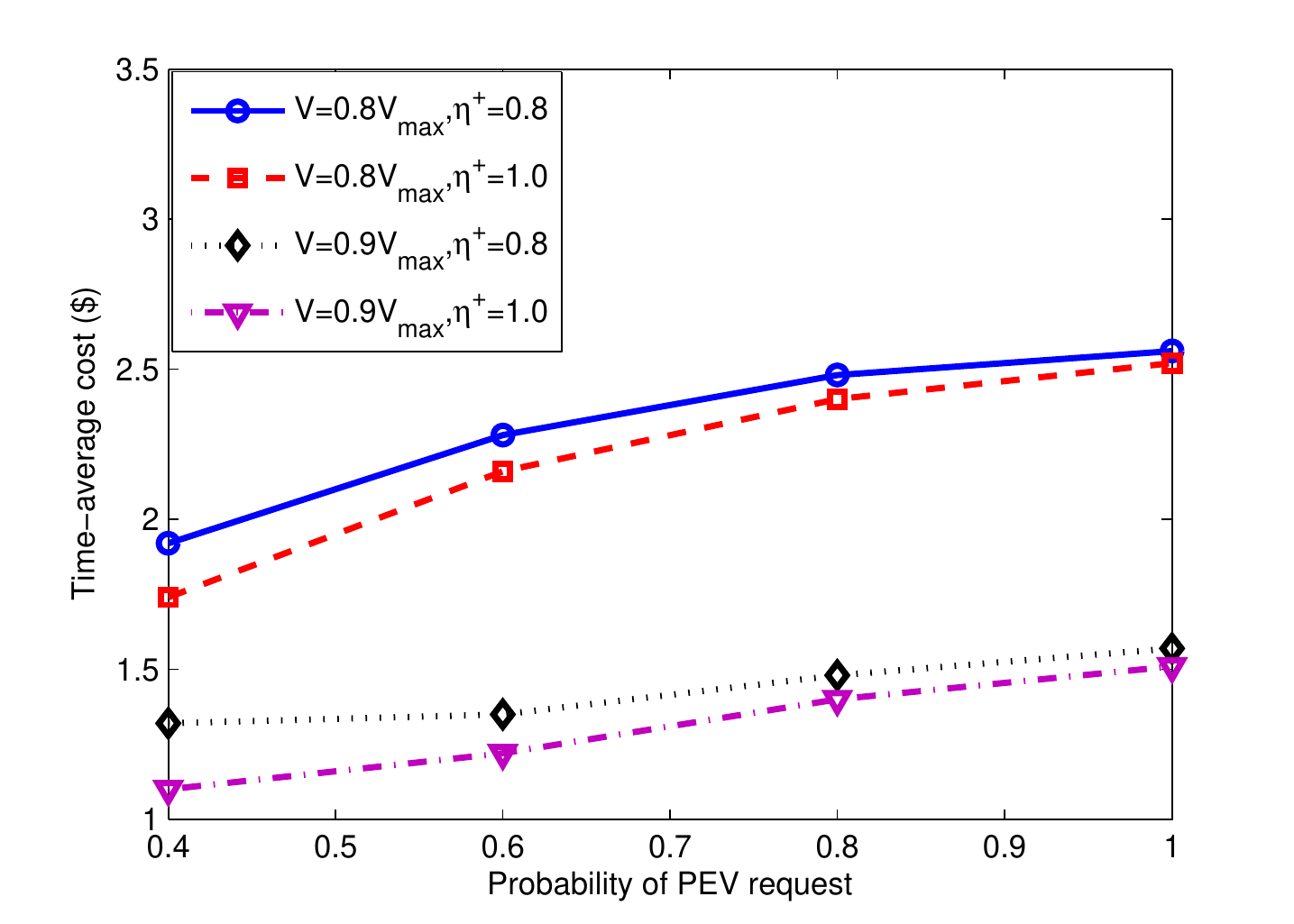}
\caption{Time-average on-grid energy cost v.s. probability of PEVs request.}\label{fig_8}
\end{figure}

For comparison, we consider the charging guidance algorithm in \cite{Ban}. Specifically, the charging guidance designed in \cite{Ban} is aimed to optimize the load balance of charging stations and minimize the system waiting time for charging. For fair comparison, we assume the following adapted algorithm consisting of the charging guidance in \cite{Ban} and the energy management in this paper. At each time slot, the adapted algorithm in \cite{Ban} is described as follows:

Given the overall charging request $\gamma$, the optimal charging request allocated to each charging station $\gamma_i$ can be calculated based on the solution to a waiting time minimization problem, which considers the coupling among charging stations. Then, the new charging request is guided to the corresponding charging station with probability $\gamma_i/\gamma$. For energy usage, the needed energy for minimizing system waiting time is supplied by renewable generator firstly, and the remaining renewable energy is stored in the battery with the feasibly largest proportion. Otherwise, if the renewable energy is insufficient, the stored energy is used. If the energy is still in shortage, the on-grid energy fills the gap.

A greedy algorithm is also proposed for comparison. At each time slot, if there is at least one idle outlet, the PEV will be guided to an idle outlet in the station that has the most residual energy in the storage device. While the charging rate is always fulfilled by the maximum value given that the distribution network constraint is satisfied.

Fig. \ref{fig_9} shows the comparison of the time-average non-renewable energy cost among the three algorithms. It is found that the system time average cost of the three algorithms decreases as enlarging battery size, since larger batteries can store more renewable energy. In general the cost achieved by the designed algorithm is the smallest no matter with different battery charging efficiency since it can utilize the time diversity of electricity price to absorb on-grid energy.

Fig. \ref{fig_10} illustrates that the average waiting time\footnote{Waiting time is defined as the period, beginning with a PEV's request accepted by an outlet and ending with its demand totally satisfied, divided by the amount of its request.} under the three algorithms. The average waiting time of the proposed algorithm grows as $B_{max}$ increasing till 500kWh since small battery can store some renewable energy and admit some PEVs into the system but its residual value is not engough to charge PEVs with full rate mostly. Considering that the energy obtained from the grid will increase the system cost, the proposed algorithm chooses to charge PEVs with lower average rate and thus longer waiting time. However, if the battery size is large enough, the average waiting time descends as battery capacity increasing. For the greedy algorithm and the adapted algorithm in \cite{Ban}, the waiting time is almost independent of the battery size, since these two algorithms charge PEVs with available residual energy in battery, otherwise it charges PEVs with on-grid energy directly. It should be noted that the waiting time of the adapted algorithm in \cite{Ban} is the smallest due to its optimal objective aiming at it. Thus, it can be concluded that the proposed algorithm can result in quicker charging service with less cost if the battery in each station is large enough.
\begin{figure}
\begin{minipage}[t]{0.5\textwidth}
\centering
\includegraphics[width=75mm]{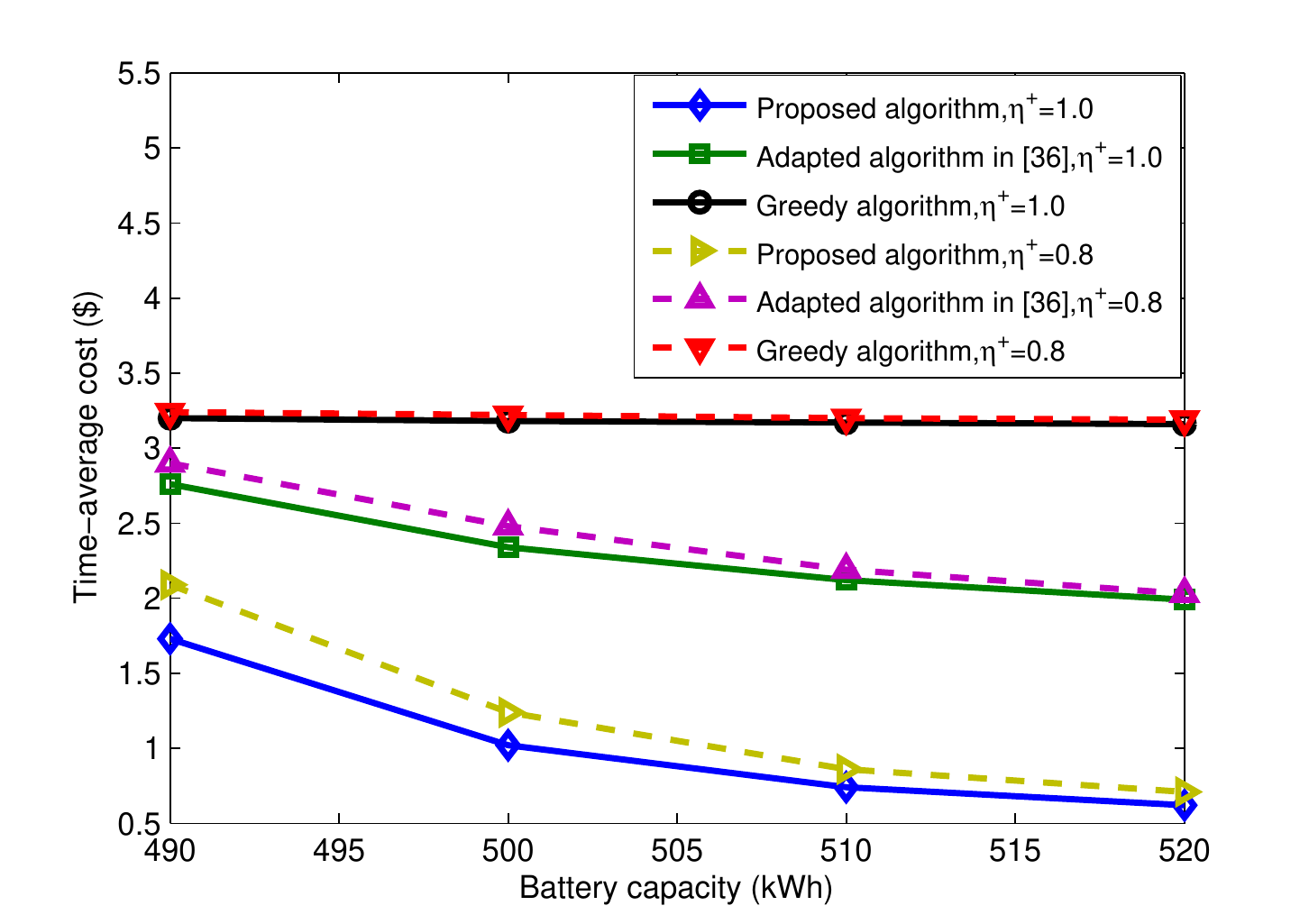}
\caption{Time-average on-grid energy cost v.s. battery capacity.}
\label{fig_9}
\end{minipage}%
\begin{minipage}[t]{0.5\textwidth}
\centering
\includegraphics[width=75mm]{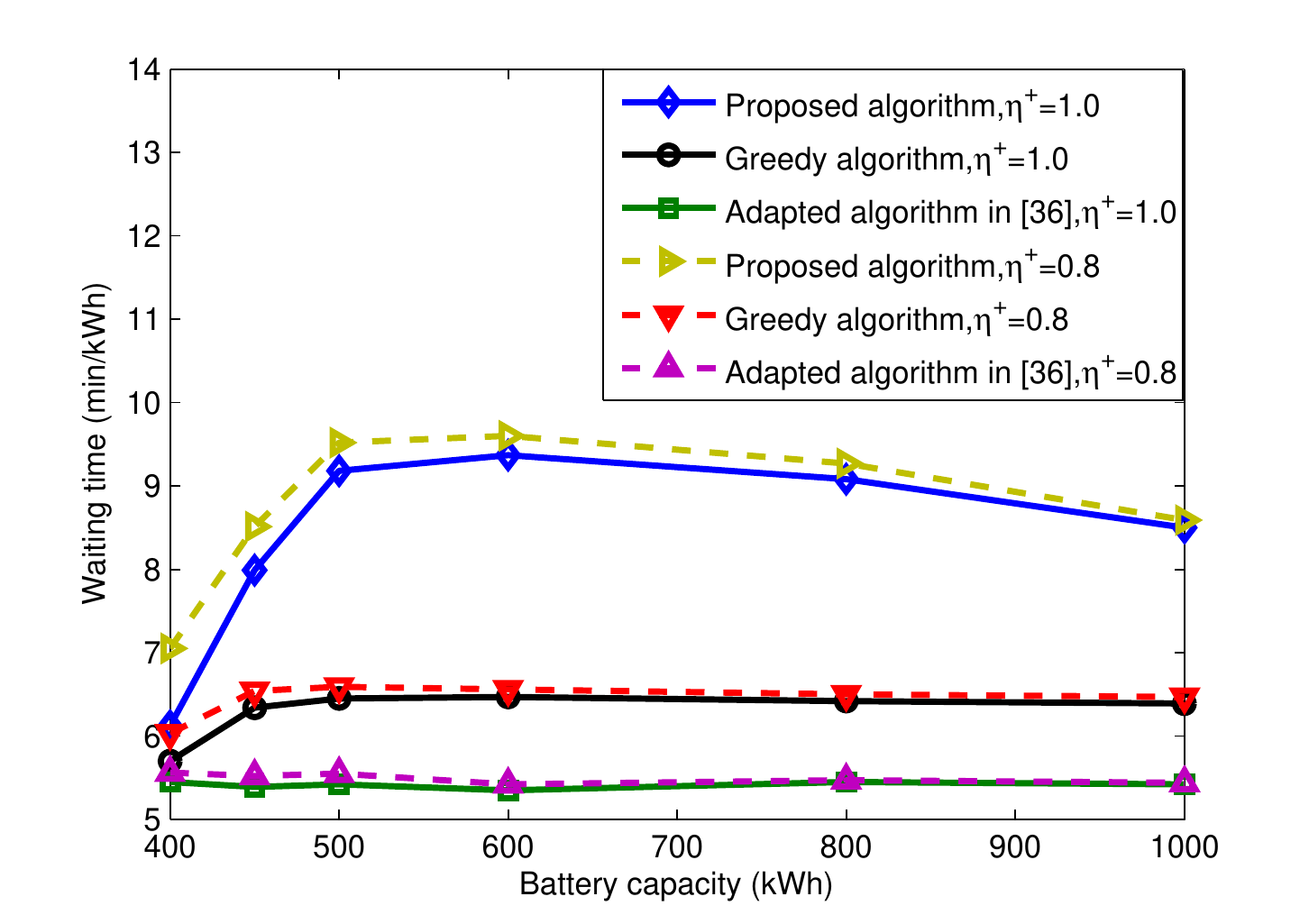}
\caption{Waiting time v.s. battery capacity.}
\label{fig_10}
\end{minipage}
\end{figure}
\begin{figure}
\begin{minipage}[t]{0.5\textwidth}
\centering
\includegraphics[width=75mm]{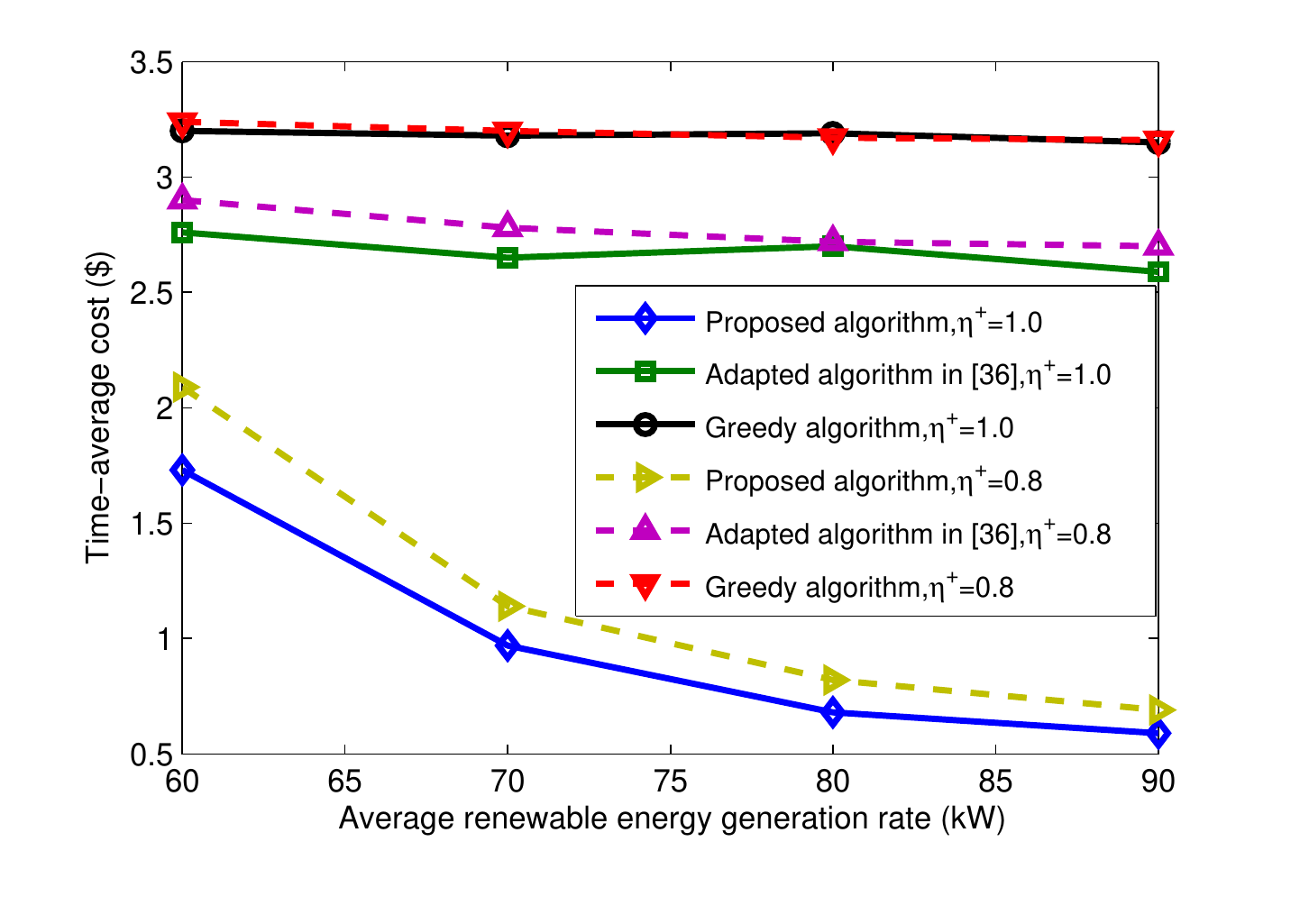}
\caption{Time-average on-grid energy cost v.s. renewable energy rate. }
\label{fig_11}
\end{minipage}%
\begin{minipage}[t]{0.5\textwidth}
\centering
\includegraphics[width=75mm]{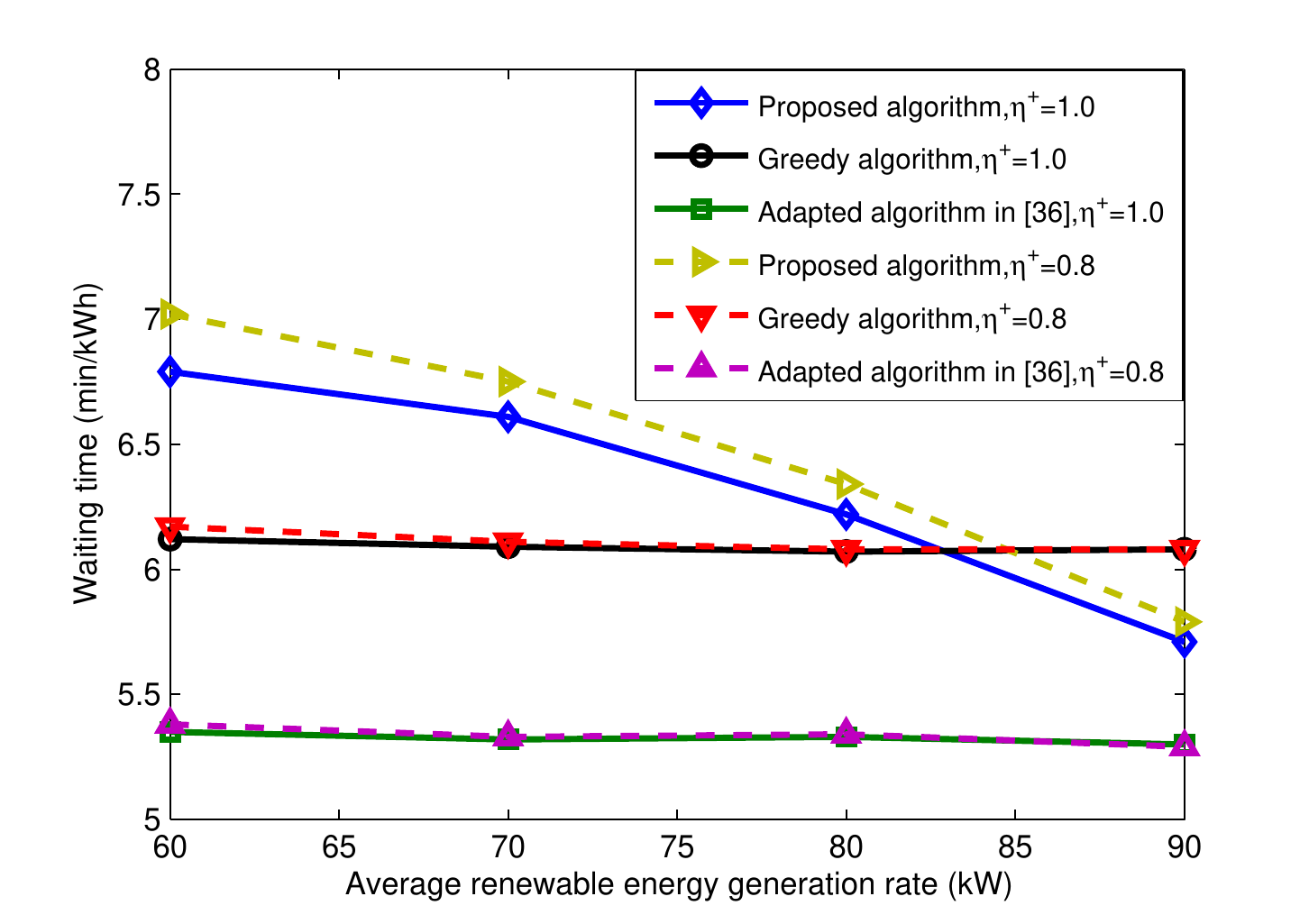}
\caption{ Waiting time v.s. renewable energy rate.}
\label{fig_12}
\end{minipage}
\end{figure}

We further evaluate the effects of average renewable energy generation rate on the time-average cost and waiting time in Fig. \ref{fig_11} and Fig. \ref{fig_12}, respectively. The similar pattern is found in Fig. \ref{fig_11} as in Fig. \ref{fig_9}, except that the time-average cost of the proposed algorithm decreases faster as the increase of renewable energy generation rate. This is due to that the generated renewable energy is utilized directly and then the residual is charged into storage for future use. Thus, more on-grid energy is saved than in the case of Fig. \ref{fig_9} for our design. Due to the same reason, it is found in Fig. \ref{fig_12} that the average waiting of our algorithm behaves better than the greedy one after a threshold generation rate.

\section{Conclusion}
This paper considers the dynamic charging problem for PEVs in a distribution network, which consists of multiple charging stations and uncontrollable non-EV loads. Each charging station can supply random arrival PEVs with energy from local renewable energy generator directly, the storage device or the distribution network. In order to cut down the on-grid energy cost and avoid overload risk in lines/transformers, Lyapunov optimization and chance constraints are utilized to develop dynamic charging guidance for PEVs and energy management scheme for stations with inefficient storage device. The proposed algorithm does not rely on the statistics of underlying processes but has provable performance. Simulation results with real wind and electricity price profiles are provided to show that the proposed algorithm can achieve lower cost with comparable waiting time compared with peer work and a greedy algorithm.
\end{spacing}
\begin{spacing}{1.5}

\end{spacing}
\end{document}